\documentclass[a4paper,english,fleqn,leqno,orivec]{llncs}
\input{packages}
\usepackage{newtxtext}

\newcommand*{\EnsSig}{\mathit{\Sigma}}

\newcommand*{\EnsFrm}[1][]{\Psi_{\ifthenelse{\isempty{#1}}{\EnsSig}{#1}}}

\newcommand*{\CEAActs}[1][]{\mathcal{A}_{\ifthenelse{\isempty{#1}}{\EnsSig}{#1}}}
\newcommand*{\CEEActs}[1][]{\mathcal{E}_{\ifthenelse{\isempty{#1}}{\EnsSig}{#1}}}

\newcommand*{\NZ}{\mathbb{N}}

\newcommand*{\powerset}{\mathnormal{\wp}}

\newcommand*{\restrict}[2]{#1\mathnormal{\upharpoonright}#2}

\newcommand*{\truefrm}{\mathrm{true}}
\newcommand*{\falsefrm}{\mathrm{false}}

\newcommand{\code}[1]{\normalfont\texttt{\spaceskip=3pt\frenchspacing\def\{{\char123 }\def\}{\char125 }\def\^{\char94}\def\_{\char95}#1}}

\newcommand{\stacked}[2][1]{\bgroup\renewcommand{\arraystretch}{#1}\begin{array}[t]{@{}l@{}}#2\end{array}\egroup}
\newcommand{\stackedtext}[1]{\text{\begin{tabular}[c]{@{}l@{}}#1\end{tabular}}}

\mathchardef\cln\mathcode`\:
\mathcode`\:=\string"8000
\begingroup \catcode`\:=\active
  \gdef:{\nobreak\mskip2mu\mathpunct{}\nonscript
    \mkern-\thinmuskip{\cln}\mskip6muplus1mu\relax}
\endgroup

\newcommand*{\Paths}{\mathscr{P}}

\newcommand*{\St}{S}
\newcommand*{\Trans}{T}
\newcommand*{\RS}[1][]{\ifthenelse{\isempty{#1}}{S_{\omega}}{S_{#1}}}
\newcommand*{\RT}[1][]{\ifthenelse{\isempty{#1}}{T_{\omega}}{T_{#1}}}
\newcommand*{\obseq}[1][O]{\mathrel{\sim_{#1}}}

\newcommand*{\interp}[2]{{#1}^{#2}}

\newcommand*{\Prop}[1][]{P\ifthenelse{\isempty{#1}}{}{_{\mathit{#1}}}}
\newcommand*{\prop}[1][]{\ifthenelse{\isempty{#1}}{p}{\mathrm{#1}}}
\newcommand*{\Ag}[1][]{A\ifthenelse{\isempty{#1}}{}{_{\mathit{#1}}}}
\newcommand*{\ag}[1][]{\ifthenelse{\isempty{#1}}{a}{\mathrm{#1}}}
\newcommand*{\ESig}[1][]{(\Prop[#1], \Ag[#1])}

\newcommand*{\esig}[1][]{\Sigma\ifthenelse{\isempty{#1}}{}{_{\mathit{#1}}}}
\newcommand*{\ebas}[1][]{\mathsf{B}\ifthenelse{\isempty{#1}}{}{_{\mathit{#1}}}}
\newcommand*{\ETS}[1][\Prop, \Ag]{\mathscr{M}_{#1}}
\newcommand*{\EMCTS}[1][\Prop, \Ag]{\mathscr{Y}_{#1}}

\newcommand*{\estr}{K}

\newcommand*{\EFrm}[1][]{\Phi_{\ifthenelse{\isempty{#1}}{\Props, \Ags}{#1}}}
\newcommand*{\DEFrm}[1][]{\Psi_{\ifthenelse{\isempty{#1}}{\Props, \Ags}{#1}}}

\newcommand*{\Props}[1][P]{#1}
\newcommand*{\Ags}[1][A]{#1}
\newcommand*{\agg}[1][]{\ifthenelse{\isempty{#1}}{G}{\mathrm{#1}}}

\newcommand*{\stt}[1][]{\ifthenelse{\isempty{#1}}{s}{\mathrm{#1}}}

\newcommand*{\CEActs}[1][]{\mathcal{A}_{\ifthenelse{\isempty{#1}}{\Props, \Ags}{#1}}}
\newcommand*{\ChEActs}[1][]{\mathcal{A}^{{?}+}_{\ifthenelse{\isempty{#1}}{\Props, \Ags}{#1}}}

\newcommand*{\ptdEU}[1][]{\mathfrak{\ifthenelse{\isempty{#1}}{u}{#1}}}

\newcommand*{\pnt}[1][]{\ifthenelse{\isempty{#1}}{q}{\mathsf{#1}}}

\newcommand*{\TEMIC}{\textsc{tEmIc}\xspace}
\newcommand*{\CTLK}{CTLK\xspace}
\newcommand*{\gact}[2]{#1 \supset #2}
\DeclareMathOperator{\nnf}{nnf}
\newcommand*{\nmodels}{\models_{\text{\normalfont n}}}
\newcommand*{\pmodels}{\models_{\text{\normalfont p}}}
\newcommand*{\cmodels}{\models}
\newcommand*{\negat}{\mathop{\boldsymbol{\nshortparallel}}}

\newcommand*{\Kop}{\mathsf{K}}
\newcommand{\K}[2]{\mathnormal{\Kop}_{#1}\!\mathop{}#2}
\newcommand*{\Mop}{\mathsf{M}}
\newcommand{\M}[2]{\mathnormal{\Mop}_{#1}\!\mathop{}#2}
\newcommand*{\Eop}{\mathsf{E}}

\newcommand*{\testact}[1]{#1{?}}

\newcommand*{\Th}[1][]{\mathit{Th}_{\ifthenelse{\isempty{#1}}{\EFrm}{#1}}}

\newcommand*{\Aop}{\mathsf{A}}
\newcommand*{\Xop}{\mathsf{X}}
\newcommand*{\Gop}{\mathsf{G}}
\newcommand*{\Fop}{\mathsf{F}}
\newcommand*{\Uop}{\mathsf{U}}
\newcommand*{\Rop}{\mathsf{R}}

\newcommand*{\EX}[1]{\mathnormal{\Eop\Xop}\!\mathop{}#1}
\newcommand*{\AX}[1]{\mathnormal{\Aop\Xop}\!\mathop{}#1}
\newcommand*{\EG}[1]{\mathnormal{\Eop\Gop}\!\mathop{}#1}
\newcommand*{\AG}[1]{\mathnormal{\Aop\Gop}\!\mathop{}#1}
\newcommand*{\EF}[1]{\mathnormal{\Eop\Fop}\!\mathop{}#1}
\newcommand*{\AF}[1]{\mathnormal{\Aop\Fop}\!\mathop{}#1}
\newcommand*{\EU}[2]{\mathnormal{\Eop}[#1\mathbin{\Uop}#2]}

\newcommand*{\AR}[2]{\mathnormal{\Aop}[#1\mathbin{\Rop}#2]}

\newcommand{\structrule}[2]{%
  \dfrac{%
    \renewcommand{\arraystretch}{1.1}%
    \begin{array}{@{}c@{}}#1\end{array}%
  }{%
    \renewcommand{\arraystretch}{1.1}%
    \begin{array}{@{}l@{}}#2\end{array}%
  }%
}

\newcommand*{\trans}[2][M]{\xrightarrow{#2}\ifthenelse{\isempty{#1}}{}{_{#1}}}
\newcommand*{\transMone}[2][M_1]{\xrightarrow{#2}\ifthenelse{\isempty{#1}}{}{_{#1}}}
\newcommand*{\transMtwo}[2][M_2]{\xrightarrow{#2}\ifthenelse{\isempty{#1}}{}{_{#1}}}

\newcommand*{\nof}[1]{\text{\ensuremath{#1}}}

\newcommand*{\LEFrm}[1][]{\mathcal{K}\ifthenelse{\equal{#1}{}}{}{_{#1}}}

\newcommand{\splitatcommas}[1]{%
  \begingroup
  \ifnum\mathcode`,="8000
  \else
    \begingroup\lccode`~=`, \lowercase{\endgroup
      \edef~{\mathchar\the\mathcode`, \penalty0 \noexpand\hspace{0pt plus .1em}}%
    }\mathcode`,="8000
  \fi
  #1%
  \endgroup
}
\newcommand*{\Seqop}{\mathrel{\Rightarrow}}
\newcommand*{\Seq}[2]{%
\ifthenelse{\equal{#1}{}\AND\equal{#2}{}}%
  {\vphantom{,\Delta}\mathnormal{\Seqop}\vphantom{,\Delta}}%
  {\ifthenelse{\equal{#1}{}}%
    {\mathnormal{\Seqop}\;\vphantom{,\Delta}\splitatcommas{#2}}%
    {\ifthenelse{\equal{#2}{}}%
      {\splitatcommas{#1}\vphantom{,\Delta}\;\mathnormal{\Seqop}}%
      {\vphantom{,\Delta}\splitatcommas{#1}\Seqop\splitatcommas{#2}\vphantom{,\Delta}}}}%
}

\newcounter{rulecounter}
\makeatletter
\NewDocumentCommand{\ruletag}{+m}{%
  \protected@edef\rulet@g{\ensuremath{(\normalfont\text{#1})}}%
  \refstepcounter{rulecounter}%
  \cref@constructprefix{rule}{\cref@result}%
  \protected@edef\@currentlabel{\rulet@g}%
  \protected@edef\@currentlabelname{\rulet@g}%
  \protected@edef\cref@currentlabel{%
    [rule][][\cref@result]%
    \rulet@g%
  }%
  \rulet@g%
}
\makeatother
\crefname{rule}{}{}
\Crefname{rule}{}{}

\newcommand{\ie}{i.\kern1pt e.\relax\xspace}
\newcommand{\eg}{e.\kern1pt g.\relax\xspace}
\newcommand{\wrt}{w.\kern1pt r.\kern1pt t.\relax\xspace}

\title{Interpreting Knowledge-based Programs (Extended Version with Proofs)\thanks{A short version without proofs has been submitted to ESOP 23.}}

\author{%
  Alexander Knapp\inst{1} \and
  Heribert Mühlberger\inst{1} \and
  Bernhard Reus\inst{2}
}
\institute{
  Universität Augsburg, Germany\\
  \email{$\{$knapp$,\;$muehlber$\}$@informatik.uni-augsburg.de}\\[.5ex]
\and 
  University of Sussex, U.K.\\
  \email{bernhard@sussex.ac.uk}
}

\begin{document}

\maketitle

\begin{abstract}
Knowledge"=based programs specify multi"=agent protocols with epistemic guards
that abstract from how agents learn and record facts or information about other
agents and the environment.  Their interpretation involves a non"=monotone
mutual dependency between the evaluation of epistemic guards over the reachable
states and the derivation of the reachable states depending on the evaluation of
epistemic guards.  We apply the technique of a must/cannot analysis invented for
synchronous programming languages to the interpretation problem of
knowledge"=based programs and demonstrate that the resulting constructive
interpretation is monotone and has a least fixed point.  We relate our approach
with existing interpretation schemes for both synchronous and asynchronous
programs.  Finally, we describe an implementation of the constructive
interpretation and illustrate the procedure by several examples and an
application to the Java memory model.
\end{abstract}

\section{Introduction}\label{sec:intro}

Knowledge"=based programs~\cite{fagin-et-al:2003} describe multi"=agent systems
based on explicit knowledge tests on what an agent knows or does not know about
itself, other agents, and the environment: Extending standard programs, an agent
may look beyond what it can directly observe by reasoning about the possible
states of the other agents and the environment in all possible program
executions.  Such non-local, epistemic conditions abstract from how an agent
may learn and record particular environmental facts or information about other
agents.
Thus knowledge"=based programs rather are specifications of (multi"=agent)
protocols that may be implemented by standard, directly executable programs.
For being implementable in the first place, however, it has to be ensured that
the knowledge guards can be resolved consistently given all possible program
executions.

Consider for example a bit transmission~\cite[Ex.~4.1.1,
Ex.~7.1.1]{fagin-et-al:2003}, where a sender $\ag[S]$ has to transmit a bit
$\prop[sbit]$ over a lossy channel to a receiver $\ag[R]$ who has to acknowledge
the reception, again over a lossy channel.
This can be modelled by a knowledge"=based program over the state variables
$\prop[sbit] \in \{ 0, 1 \}$, $\prop[rval] \in \{ \bot, 0, 1 \}$, and
$\prop[ack] \in \{ 0, 1 \}$ as follows: $\ag[S]$ can only directly observe
(read) $\prop[sbit]$ and $\prop[ack]$, and $\ag[R]$ only $\prop[rval]$ (but both
may write all variables); $(\K{\ag[R]}{\prop[sbit] = 0}) \lor
(\K{\ag[R]}{\prop[sbit] = 1})$ expresses that $\ag[R]$ knows $\prop[sbit]$'s
value and is abbreviated by $\K{\ag[R]}{\mathit{sbit}}$.  The behaviour
description consists of a looping guarded command with two branches (where
\code{or} means a non"=deterministic choice and \code{skip} doing nothing) that
is started with $\prop[rval] = \bot$ and $\prop[ack] = 0$, but $\prop[sbit]$
left undetermined:
%
\begin{equation*}
\begin{array}[t]{@{}r@{\ }l@{\ }l@{\quad}l@{}}
\textbf{\texttt{do}} & \neg\K{\ag[S]}{\K{\ag[R]}{\mathit{sbit}}} \rightarrowtriangle (\prop[rval] \gets \prop[sbit]\ \code{or}\ \code{skip}) & & \text{--- }\ag[S]\\
\talloblong & \K{\ag[R]}{\mathit{sbit}} \land \neg\K{\ag[R]}{\K{\ag[S]}{\K{\ag[R]}{\mathit{sbit}}}} \rightarrowtriangle (\prop[ack] \gets 1\ \code{or}\ \code{skip}) & \textbf{\texttt{od}} & \text{--- }\ag[R]
\end{array}
\end{equation*}
The epistemic formulæ $\K{\ag}{\varphi}$ in the program are to be interpreted as
in classical Kripke semantics: $\varphi$ holds in all states (or worlds) that
agent $\ag$ currently deems possible.  Which states these are is regulated on
the one hand by what state information $\ag$ can observe: if the current state
is indistinguishable from another on account of these available observations,
both are possible for the agent.  In the example above only $\ag[S]$ can
observe $\prop[sbit]$, though, due to the protocol, it should be possible that
eventually $\ag[R]$ knows its value.  On the other hand, the possible states
depend on which runs of the knowledge"=based program may actually happen, \ie,
which states are reachable when evaluating the epistemic guards and taking
transitions: If only the actions of the program are taken, it is impossible to
reach a state satisfying both $\prop[rval] \neq \bot$ and $\prop[rval] \neq
\prop[sbit]$, which, however, is present in the global state space; but it is
decisive that it is not reachable in any execution in order to have some
execution where $\K{\ag[R]}{\mathit{sbit}}$ can become true.  Conversely, there
is an execution where $\K{\ag[R]}{\mathit{sbit}}$ (and
$\K{\ag[S]}{\K{\ag[R]}{\mathit{sbit}}}$) remains false.

The interpretation of knowledge"=based programs  hinges precisely on this mutual
dependency between the evaluation of epistemic guards over the reachable states
and the derivation of the reachable states depending on the evaluation of the
epistemic guards. This implicit definition of the epistemic state of the agents
by the observables and the reachable states of the commonly known protocol is in
stark contrast to Baltag's epistemic action
models~\cite{baltag-moss:synthese:2004,van-ditmarsch-van-der-hoek-kooi:2008},
where the epistemic state is given and manipulated explicitly.  In many cases,
including the bit transmission protocol, the reachable state space may be
computed using static analysis techniques without taking into account the
epistemic nature of the guards.  However, the interplay between knowledge and
reachability may sometimes become more intricate: The more states are reachable
the less is known definitely, and the guards will in turn influence what is
reachable positively or negatively.

Consider, for another example, a variable setting
problem~\cite[Exc.~7.5]{fagin-et-al:2003} involving a single agent $\ag[a]$ and
a single state variable $\prop[x] \in \{ 0, 1, 2, 3 \}$, where $\ag[a]$ cannot
observe $\prop[x]$ directly.  The agent executes the following guarded command
starting with $\prop[x] = 0$:
\begin{equation*}
\begin{array}[t]{@{}r@{\ }l@{\ }r@{}}
  \textbf{\texttt{if}} & \K{\ag[a]}{\prop[x] \neq 1} \rightarrowtriangle \prop[x] \gets 3
\\
  \talloblong          & \K{\ag[a]}{\prop[x] \neq 3} \rightarrowtriangle \prop[x] \gets 1 & \textbf{\texttt{fi}}
\end{array}
\end{equation*}
Clearly, $\prop[x] = 0$ definitely is reachable and $\prop[x] = 2$ definitely is
not reachable.  However, two different sets of reachable states make for a
consistent interpretation of the knowledge guards for the remaining values: $\{
\prop[x] = 0, \prop[x] = 1 \}$, where $\K{\ag[a]}{\prop[x] \neq 1}$ is false and
$\K{\ag[a]}{\prop[x] \neq 3}$ is true, and $\{ \prop[x] = 0, \prop[x] = 3 \}$,
with the opposite results.  The singleton set $\{ \prop[x] = 0 \}$ is ruled out,
since both guards would be true such that $\prop[x] = 3$ and $\prop[x] = 1$ are
reachable; and $\{ \prop[x] = 0, \prop[x] = 1, \prop[x] = 3 \}$ is impossible,
since both guards are false and thus neither $\prop[x] = 1$ nor $\prop[x] = 3$
are reachable.
Breaking this cycle by making one of the transitions unconditional on knowledge as, \eg, in
\begin{equation*}
\begin{array}[t]{@{}r@{\ }l@{\ }l@{}}
  \textbf{\texttt{if}} & \K{\ag[a]}{\prop[x] \neq 1} \rightarrowtriangle \prop[x] \gets 3\\
  \talloblong & \K{\ag[a]}{\prop[x] \neq 3} \rightarrowtriangle \prop[x] \gets 2\\
  \talloblong & \truefrm \rightarrowtriangle \prop[x] \gets 1 & \textbf{\texttt{fi}}
\end{array}
\end{equation*}
yields a knowledge"=based program with the unique solution $\{ \prop[x] = 1,
\prop[x] = 2 \}$.  For computing its behaviour, however, several steps are
needed, first reasoning that $\prop[x] = 1$ is reachable, then that $\prop[x] =
3$ is not reachable, and, finally, that $\prop[x] = 2$ is reachable.

\paragraph{Related Work.}
In their introduction and seminal treatise on knowledge"=based
programs~\cite{fagin-et-al:dc:1997,fagin-et-al:2003}, Fagin et al.\ characterise
the unique interpretability of such programs by their ``dependence on the past''
\wrt some non-empty class of transition systems: It has to be ensured that the
evaluation of knowledge guards in a state coincides for all interpretations in
the class that share a common past of the current state.  A sufficient condition
for dependence on the past is that the program ``provides epistemic witnesses''
for all interpretations of the class such that not knowing something at some
point in time has a counter example in the past.  A sufficient condition for the
provision of epistemic witnesses, in turn, is that the program is
``synchronous'', \ie, that all agents can determine the global time from their
local states.  For example, the bit transmission protocol does provide epistemic
witnesses and thus is uniquely interpretable; but it is not synchronous.  The
cycle-breaking variable setting program is also uniquely interpretable, but does
not provide epistemic witnesses and is thus not synchronous, despite having only
one agent.  For ``asynchronous'' knowledge"=based programs, De Haan et
al.~\cite{de-haan-et-al:fundinf:2004} suggest to rely on classical iteration of
the non"=monotone reachability functional that interprets the knowledge
modalities according to what currently is assumed to be reachable.  The
computation process is started with all states assumed to be reachable and stops
when some set of states is repeated.  This approach fixes some semantics for all
knowledge"=based programs, also for those which are cyclic and contradictory or
only self-fulfilling.

The problem of mutual dependence of guard evaluation and reachability has also
occurred in the design of synchronous programming
languages~\cite{beneviste-et-al:ieee:2003} for embedded systems, like
Esterel~\cite{berry:milner:2000} or Lustre~\cite{halbwachs-et-al:ieee:1991},
which rely on ``perfect synchrony'': a step for reacting to some inputs takes
zero time and output signals are produced at exactly the same time as the input
signals.  Since thus the status of a signal to be produced can be queried at the
same time, this requires ``logical coherence'' saying that a (non-input) signal
is present in a step of execution if, and only if, a command emitting this
signal is executed in this step.  Whereas Lustre forbids cyclic programs on a
syntactic basis, Berry's approach to the semantics of
Esterel~\cite{berry:esterel} singles out ``reactive'' --- at least one execution
--- and ``determinate'' --- at most one execution --- programs using a static
executability analysis: It is computed which signals \emph{must} be present,
\ie, have to occur inevitably, and which signals \emph{cannot} be present, \ie,
have no emitting execution.  This is also referred to as must/cannot analysis
and has to be performed several times for finding a fixed point of all the
signal statuses.  Notably, it does not require witnesses in the past.

\paragraph{Contributions.}
We apply the principles of the must/cannot analysis to the interpretability
problem of knowledge"=based programs.  After recalling some basic notions of
epistemic logic and epistemic transition structures (\cref{sect:epistemic}), we
first recapitulate the approaches by Fagin et al.~\cite{fagin-et-al:2003} and De
Haan et al.~\cite{de-haan-et-al:fundinf:2004} in terms of epistemically guarded
transition systems, a syntax-agnostic format for knowledge"=based programs
(\cref{sect:kbp}).  In contrast to those designs our account is state"=based
rather than run"=based which offers a more direct analysis.  We demonstrate the
results and the limits of both interpretation schemes by several examples that
illustrate (a"=)synchronicity and non"=monotone interpretation for cyclic,
contradictory, or self"=fulfilling programs.  The latter behaviour is the main
motivation for our reformulation of the interpretation problem in terms of
epistemic must/can transition structures which offer lower and upper bounds on
the behaviour of a knowledge"=based program (\cref{sect:interp}).  We show that
this constructive interpretation is always monotone and yields a least fixed
point.  However, lower and upper bound of the fixed point need not always
coincide and we relate decided fixed points with the notions of ``providing
epistemic witnesses'' and synchronicity.  We then generalise the must/can
approximation technique to inference systems with positive and negative
premisses (\cref{sect:rules}).  The behaviour of a knowledge"=based program can
be represented by a rule system that also allows to express that some fact
cannot be deduced; the monotonicity and fixed point properties can be directly
transferred to these general rule systems.  We finally describe an
implementation of our constructive interpretation approach in the ``Temporal
Epistemic Model Interpreter and Checker'' (\TEMIC, \cref{sect:temic}).  For
model checking interpreted knowledge"=based programs, the tool supports \CTLK,
the combination of ``Computational Tree Logic'' (CTL) with epistemic logic.
Moreover, this logic can also be used in program guards; the interpretation of
such temporal"=epistemic programs extends the previous approaches.  We give some
applications to the analysis of the Java memory model.

\section{Epistemic Logic and Epistemic Transition Structures}\label{sect:epistemic}

We briefly summarise the basic notions of epistemic logic for expressing
knowledge
guards~\cite{van-ditmarsch-van-der-hoek-kooi:2008,van-ditmarsch-et-al:hel:2015}.
We then define epistemic transition structures as the domain of interpretation
of knowledge"=based programs.  These transition structures combine the temporal
dimension of executing a program with the epistemic dimension for evaluating
what agents know.  Both the logic and the transition structures are built over
an \emph{epistemic signature} $\esig = \ESig$ that consists of a set of
\emph{propositions} $\Prop$ and a set of \emph{agents} $\Ag$.

\subsection{Epistemic Logic}

An \emph{epistemic structure} $K = (W, R, L)$ over $\ESig$ is given by a set of
\emph{worlds} $W$, an $\Ag$-family of epistemic \emph{accessibility relations}
$R = (R_{\ag} \subseteq W \times W)_{\ag \in \Ag}$, and a \emph{labelling} $L : W
\to \powerset \Prop$.  In concrete examples, we will require $R_{\ag}$ to be an
equivalence relation such that if $(w_1, w_2) \in R_{\ag}$, then agent $\ag$
cannot distinguish between the two worlds $w_1$ and $w_2$.  The \emph{epistemic
  formulæ} $\varphi \in \EFrm$ over $\ESig$ are defined by the following
grammar:
\begin{align*}
  \varphi
&{\;\cln\cln=\;}\begin{array}[t]{@{}l@{}}
   \prop \;\mid\; \falsefrm \;\mid\; \neg\varphi \;\mid\; \varphi_1 \land \varphi_2\;\mid\; \K{\ag}{\varphi}
\end{array}
\end{align*}
where $\prop \in \Prop$ and $\ag \in \Ag$.  The epistemic formula
$\K{\ag}{\varphi}$ is to be read as ``agent $\ag$ \emph{knows $\varphi$}''.  We
use the usual propositional abbreviations $\truefrm$ for $\neg\falsefrm$ and
$\varphi_1 \lor \varphi_2$ for $\neg(\neg\varphi_1 \land \neg\varphi_2)$.
Furthermore, we consider the epistemic modality $\Mop$ as the dual of $\Kop$,
such that $\M{\ag}{\varphi}$ abbreviates $\neg\K{\ag}{\neg\varphi}$ and is to be
read as ``agent $\ag$ \emph{deems $\varphi$ possible}''.  The \emph{satisfaction
  relation} of an epistemic formula $\varphi \in \EFrm$ over an epistemic
structure $K = (W, R, L)$ over $\ESig$ at a world $w \in W$, written $K, w
\models \varphi$, is inductively defined by
\begin{align*}
  K, w &\models \prop \iff \prop \in L(w)
\\
  K, w &\not\models \falsefrm
\\
  K, w &\models \neg\varphi \iff K, w \not\models \varphi
\\
  K, w &\models \varphi_1 \land \varphi_2 \iff \text{$K, w \models  \varphi_1$ and $K, w \models \varphi_2$}
\\
  K, w &\models \K{\ag}{\varphi} \iff K, w' \models \varphi \text{ f.\,a. } w' \in W \text{ with } (w, w') \in R_{\ag}
\end{align*}

\subsection{Epistemic Transition Structures}\label{sect:ets}

An epistemic transition structure combines a temporal transition relation with an
epistemic accessibility relation over a common set of states.  The transitions
describe which states can be reached from a set of initial states, the
accessibilities which states are indistinguishable.  Knowledge formulæ are
evaluated over the associated global epistemic structure.  This derived
structure has the reachable states as its worlds and reuses the accessibility
relation but restricted to the reachable states.

Formally, an \emph{epistemic transition structure} $M = (S, E, L, S_0, T)$ over
$\ESig$ is given by an epistemic structure $(S, E, L)$, a set of temporally
\emph{initial states} $S_0 \subseteq S$, and a temporal \emph{transition
  relation} $T \subseteq S \times S$.  We write $\St(M)$ for $S$, $\Trans(M)$
for $T$, etc.  The (temporally) \emph{reachable states} $\RS(M) = \bigcup_{0
  \leq k} \RS[k](M)$ and \emph{transition relation} $\RT(M) = \bigcup_{0 \leq k}
\RT[k](M)$ of $M$ are inductively defined by
\begin{gather*}
\RS[0](M) = S_0
\text{,}\quad
\RS[k+1](M) = \RS[k](M) \cup \{ s' \mid \text{ex.\ $s \in \RS[k](M)$ s.\,t.\ $(s, s') \in T$} \}
\\
\RT[0](M) = \emptyset
\text{,}\quad
\RT[k+1](M) = \RT[k](M) \cup \{ (s, s') \in T \mid s \in \RS[k](M) \}
\ \text{.}
\end{gather*}
The associated \emph{epistemic structure} of $M$ is given by $\estr(M) =
(\RS(M),\allowbreak E \cap \RS(M)^2,\allowbreak \restrict{L}{\RS(M)})$.  The
\emph{satisfaction relation} of an epistemic formula $\varphi \in \EFrm$ over
$M$ at an $s \in \RS(M)$, written $M, s \models \varphi$, is defined as
\begin{equation*}
M, s \models \varphi \iff \estr(M), s \models \varphi
\ \text{.}
\end{equation*}
The set of epistemic transition structures over $\esig = \ESig$ sharing the same
\emph{epistemic state basis} $\ebas = (S,\allowbreak E,\allowbreak L,\allowbreak
S_0)$ is denoted by $\ETS[\esig](\ebas)$.  We say that $M_1 \subseteq M_2$ for
$M_1, M_2 \in \ETS[\esig](\ebas)$ if $\Trans(M_1) \subseteq \Trans(M_2)$ and
similarly extend union and intersection from transition relations to epistemic
transition structures.


\section{Knowledge-based Programs}\label{sect:kbp}

Knowledge"=based programs extend standard programs by explicit knowledge tests.
Their interpretation involves a cycle: the evaluation of the epistemic guards
depends on the program's reachable states, the derivation of the reachable
states on the evaluation of the program's epistemic guards.

We render knowledge"=based programs in a syntax-agnostic format as epistemically
guarded transition systems.  Like epistemic transition structures, these guarded
systems operate on a global set of states with epistemic accessibilities and a
propositional labelling.  All program steps are represented as
knowledge"=guarded actions of the form $\gact{\varphi}{B}$ with $\varphi$ an
epistemic formula and $B$ a relation on the semantic states.
Knowledge"=independent decisions are obtained by choosing $\varphi = \truefrm$,
and any kind of program control structure can be expressed by a judicious choice
of guarded actions.

Breaking up the cyclic step of assigning meaning to a knowledge"=based program,
an epistemically guarded transition system $\Gamma$ is interpreted over an
epistemic transition structure $M$ yielding another epistemic transition
structure $\Gamma^M$.  A guarded action $\gact{\varphi}{B}$ of $\Gamma$
contributes those $(s, s') \in B$ for which $M, s \models \varphi$, where, in
particular, $s$ is reachable in $M$.  What is sought for is a consistent
interpretation with $\Gamma^M = M$ such that reachability and knowledge are
mutually justified.  Finding such a balanced structure is complicated by the
fact that the interpretation functional is not monotone in general: The more is
reachable the less is known and this may make more or less states reachable.

After introducing and illustrating our format of knowledge"=based programs we
summarise and adapt two existing approaches to their interpretation that have
been devised for run"=based rather than state"=based systems: De Haan et
al.~\cite{de-haan-et-al:fundinf:2004} propose to iterate the interpretation
functional starting from an epistemic transition structure where all states are
reachable.  Iteration stops when either a fixed point is reached or, due to
non-monotonicity, a contradiction is found.  In this way all knowledge"=based
programs are assigned some semantics and there is no distinction between
meaningful and contradictory or just self"=fulfilling programs.  The original
approach by Fagin et al.~\cite{fagin-et-al:dc:1997,fagin-et-al:2003}
characterises knowledge"=based programs that admit a unique consistent
interpretation by the notion of dependence on the past.  A sufficient condition
of providing epistemic witnesses is developed which, in particular, applies to
the subclass of synchronous knowledge"=based programs.

\subsection{Epistemically Guarded Transition Systems}

An \emph{epistemically guarded transition system} $\Gamma = (S, E, L, S_0,
\mathcal{T})$ over $\ESig$ is given by an epistemic state basis $(S, E, L, S_0)$
over $\ESig$ and a set $\mathcal{T}$ of \emph{epistemically guarded actions}
$\gact{\varphi}{B}$ consisting of an epistemic formula $\varphi \in \EFrm$ as
\emph{guard} and a transition relation $B \subseteq S \times S$.

\begin{example}\label{ex:egts}
\begin{refpars}
  \item\label{it:ex:egts:bt} Consider the bit transmission problem of the
introduction:
\newcommand*{\Kr}{\mathsf{k}_{\mathit{r}}}
\newcommand*{\Ksr}{\mathsf{k}_{\mathit{sr}}}
\newcommand*{\Krsr}{\mathsf{k}_{\mathit{rsr}}}
\begin{equation*}
\begin{array}[t]{@{}r@{\ }l@{\ }l@{}}
\textbf{\texttt{do}} & \neg\K{\ag[S]}{\K{\ag[R]}{\mathit{sbit}}} \rightarrowtriangle (\prop[rval] \gets \prop[sbit]\ \code{or}\ \code{skip}) &\\
\talloblong & \K{\ag[R]}{\mathit{sbit}} \land \neg\K{\ag[R]}{\K{\ag[S]}{\K{\ag[R]}{\mathit{sbit}}}} \rightarrowtriangle (\prop[ack] \gets 1\ \code{or}\ \code{skip}) & \textbf{\texttt{od}}
\end{array}
\end{equation*}
A sender agent $\ag[S]$ sends a bit $\prop[sbit] \in \{ 0, 1 \}$ to a receiver
agent $\ag[R]$ over an unreliable channel by setting $\prop[rval] \in \{ \bot,
0, 1 \}$; and $\ag[R]$ acknowledges the reception over an unreliable channel by
setting $\prop[ack] \in \{ 0, 1 \}$.  Again, we abbreviate
$(\K{\ag[R]}{\neg\prop[sbit]}) \lor (\K{\ag[R]}{\prop[sbit]})$ expressing that
the receiver knows the bit to be sent by $\K{\ag[R]}{\mathit{sbit}}$.  We
concretise the problem into an epistemically guarded transition system
$\Gamma_{\mathit{bt}} = (\ebas_{\mathit{bt}}, \mathcal{T}_{\mathit{bt}})$ with
$\ebas_{\mathit{bt}} = (S_{\mathit{bt}}, E_{\mathit{bt}}, L_{\mathit{bt}},
S_{\mathit{bt}, 0})$ over $\esig[bt] = \ESig[bt]$ with $\Prop[bt] = \{
\prop[sbit], \prop[rbit], \prop[snt], \prop[ack] \}$ and $\Ag[bt] = \{ \ag[S],
\ag[R] \}$.  Since we use a propositional encoding, we represent $\prop[rval]
\in \{ \bot, 0, 1 \}$ by a proposition $\prop[rbit]$ for the transmitted bit and
a proposition $\prop[snt]$ for the validity of $\prop[rbit]$.  Further
abbreviating the knowledge guards $\K{\ag[R]}{\mathit{sbit}}$ by $\Kr$,
$\K{\ag[S]}{\K{\ag[R]}{\mathit{sbit}}}$ by $\Ksr$, and
$\K{\ag[R]}{\K{\ag[S]}{\K{\ag[R]}{\mathit{sbit}}}}$ by $\Krsr$, the
transition system $\Gamma_{\mathit{bt}}$ is graphically given by
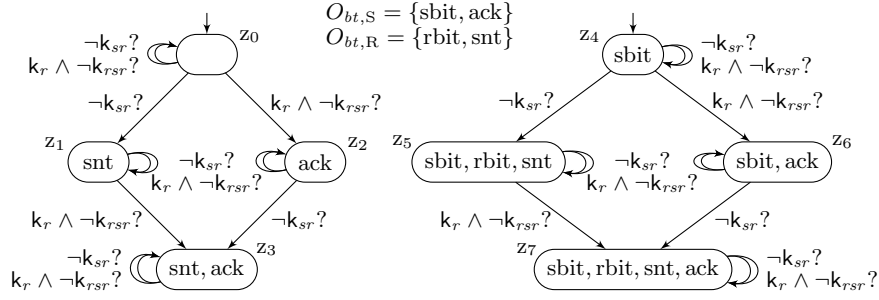
\begin{center}
\begin{tikzpicture}[scale=.95, transform shape, every loop/.style={-latex'}, every label/.style={label distance=0cm, inner ysep=0pt}, font={\fontsize{9pt}{9pt}\selectfont}]
\tikzstyle{state}=[rounded rectangle, draw, outer sep=0pt, inner xsep=4pt, inner ysep=0pt, minimum height=16pt, minimum width=32pt]
\tikzstyle{transition}=[draw,-latex']
\tikzstyle{annotation}=[outer sep=0pt, inner sep=0pt, font={\fontsize{9pt}{9pt}\selectfont}]
\node[state, label={[xshift=2pt, anchor=west]above right:$\stt[z_0]$}] (g0) at (-1.5, 0) {};
\node[state, label={[xshift=-2pt, anchor=east]above left:$\stt[z_1]$}] (g1) at (-3.0, -1.5) {$\prop[snt]$};
\node[state, label={[xshift=2pt, anchor=west]above right:$\stt[z_2]$}] (g2) at (0, -1.5) {$\prop[ack]$};
\node[state, label={[xshift=8pt, anchor=west]above right:$\stt[z_3]$}] (g3) at (-1.5, -3.0) {$\prop[snt], \prop[ack]$};
\node[state, label={[xshift=-2pt, anchor=east]above left:$\stt[z_4]$}] (g0') at (4.4, 0) {$\prop[sbit]$};
\node[state, label={[xshift=-18pt, anchor=east]above left:$\stt[z_5]$}] (g1') at (2.4, -1.5) {$\prop[sbit], \prop[rbit], \prop[snt]$};
\node[state, label={[xshift=10pt, anchor=west]above right:$\stt[z_6]$}] (g2') at (6.4, -1.5) {$\prop[sbit], \prop[ack]$};
\node[state, label={[xshift=-26pt, anchor=east]above left:$\stt[z_7]$}] (g3') at (4.4, -3.0) {$\prop[sbit], \prop[rbit], \prop[snt], \prop[ack]$};
\coordinate (W0) at ($ (g0.north) + (0, 0.3) $);
\coordinate (W0') at ($ (g0'.north) + (0, 0.3) $);
%
\node[annotation, align=left] at ($ 0.5*(g0) + 0.5*(g0') + (0, 0.4) $) {$O_{\mathit{bt}, \ag[S]} = \{ \prop[sbit], \prop[ack] \}$\\[.25ex] $O_{\mathit{bt}, \ag[R]} = \{ \prop[rbit], \prop[snt] \}$};
\path[transition]
  (W0) -- (g0)
  (g0) edge
       node[anchor=east, yshift=2pt] {$\testact{\neg\Ksr}$}
  (g1)
  (g0) edge
       node[anchor=west, yshift=2pt] {$\testact{\Kr \land \neg\Krsr}$}
  (g2)
  (g1) edge
       node[anchor=east, yshift=-2pt] {$\testact{\Kr \land \neg\Krsr}$}
  (g3)
  (g2) edge
       node[anchor=west, yshift=-2pt] {$\testact{\neg\Ksr}$}
  (g3)
  (g0) edge[loop left, distance=12] (g0) edge[loop left, distance=16]
       node[anchor=east, align=right] {$\testact{\neg\Ksr}$\\[0pt]$\testact{\Kr \land \neg\Krsr}$}
  (g0)
  (g1) edge[loop right, distance=12] (g1) edge[loop right, distance=16]
       node[anchor=west, align=center, xshift=-6pt, yshift=-4pt] {$\testact{\neg\Ksr}$\\[0pt]$\testact{\Kr \land \neg\Krsr}$}
  (g1)
  (g2) edge[loop left, distance=12] (g2) edge[loop left, distance=16]
  (g2)
  (g3) edge[loop left, distance=12] (g3) edge[loop left, distance=16]
       node[anchor=east, align=right] {$\testact{\neg\Ksr}$\\[0pt]$\testact{\Kr \land \neg\Krsr}$}
  (g3)
;
\path[transition]
  (W0') -- (g0')
  (g0') edge
        node[anchor=east, yshift=2pt] {$\testact{\neg\Ksr}$}
  (g1')
  (g0') edge
        node[anchor=west, yshift=2pt] {$\testact{\Kr \land \neg\Krsr}$}
  (g2')
  (g1') edge
        node[anchor=east, yshift=-2pt] {$\testact{\Kr \land \neg \Krsr}$}
  (g3')
  (g2') edge
        node[anchor=west, yshift=-2pt] {$\testact{\neg \Ksr}$}
  (g3')
  (g0') edge[loop right, distance=12] (g0') edge[loop right, distance=16]
        node[anchor=west, align=left] {$\testact{\neg\Ksr}$\\[0pt]$\testact{\Kr \land \neg\Krsr}$}
  (g0')
  (g1') edge[loop, out=8, in=-8, distance=12] (g1') edge[loop, out=8, in=-8, distance=16]
        node[anchor=west, align=center, xshift=-6pt, yshift=-4pt] {$\testact{\neg\Ksr}$\\[0pt]$\testact{\Kr \land \neg\Krsr}$}
  (g1')
  (g2') edge[loop, out=188, in=172, distance=12] (g2') edge[loop, out=188, in=172, distance=16]
  (g2')
  (g3') edge[loop, out=8, in=-8, distance=12] (g3') edge[loop, out=8, in=-8, distance=16]
        node[anchor=west, align=left] {$\testact{\neg\Ksr}$\\[0pt]$\testact{\Kr \land \neg\Krsr}$}
  (g3')
;
\end{tikzpicture}
\end{center}
The states $S_{\mathit{bt}}$ comprise of $\{ \stt[z_0], \stt[z_1], \ldots,
\stt[z_7] \}$ with $L_{\mathit{bt}}(\stt[z_0]) = \emptyset$,
$L_{\mathit{bt}}(\stt[z_1]) = \{ \prop[snt] \}$, \ldots,
$L_{\mathit{bt}}(\stt[z_7]) = \{ \prop[sbit], \prop[rbit], \prop[snt],
\prop[ack] \}$ as outlined in the graph above; the set of initial states is
$S_{\mathit{bt}, 0} = \{ \stt[z_0], \stt[z_4] \}$.  The epistemic
accessibility relations $E_{\mathit{bt}, \ag}$ for $\ag \in \Ag[bt]$ are given
by \emph{observability sets} $O_{\mathit{bt}, \ag}$ that declare two states
$s_1, s_2 \in S_{\mathit{bt}}$ to be \emph{$O_{\mathit{bt},
    \ag}$-indistinguishable}, written as $s_1 \obseq[O_{\mathit{bt}, \ag}] s_2$,
if for all $p \in O_{\mathit{bt}, \ag}$ it holds that $p \in
L_{\mathit{bt}}(s_1) \iff p \in L_{\mathit{bt}}(s_2)$, and consequently
$E_{\mathit{bt}, \ag} = {\obseq[O_{\mathit{bt}, \ag}]}$, such that
$E_{\mathit{bt}, \ag}$ forms an equivalence relation.  Due to $\prop[sbit]
\notin O_{\mathit{bt}, \ag[R]}$, the receiver $\ag[R]$ cannot ``see''
$\prop[sbit]$ and hence cannot distinguish between states $\stt[z_0]$ and
$\stt[z_4]$, but $\ag[S]$ can.  On the other hand, $\ag[R]$ can distinguish
between $\stt[z_1]$ and $\stt[z_5]$ as $\ag[R]$ has access to $\prop[rbit]$.
Finally, $\mathcal{T}_{\mathit{bt}}$ consists of two epistemically guarded
actions
\begin{gather*}
\gact{\neg \K{\ag[S]}{\K{\ag[R]}{\mathit{sbit}}}}{\{ (\stt[z_{\nof{i}}], \stt[z_{\nof{i}}]) \mid 0 \leq i \leq 7 \} \cup \{ (\stt[z_0], \stt[z_1]), (\stt[z_2], \stt[z_3]), (\stt[z_4], \stt[z_5]), (\stt[z_6], \stt[z_7]) \}}
\ \text{,}\\
\gact{\K{\ag[R]}{\mathit{sbit}} \land \neg \K{\ag[R]}{\K{\ag[S]}{\K{\ag[R]}{\mathit{sbit}}}}}{\stacked{\{ (\stt[z_{\nof{i}}], \stt[z_{\nof{i}}]) \mid 0 \leq i \leq 7 \} \cup{}\\ \{ (\stt[z_0], \stt[z_2]), (\stt[z_1], \stt[z_3]), (\stt[z_4], \stt[z_6]), (\stt[z_5], \stt[z_7]) \}\ \text{,}}}
\end{gather*}
which directly reflect the sending and acknowledging actions of the bit
transmission problem: The system can only advance from $\stt[z_0]$ to
$\stt[z_1]$ (and $\stt[z_4]$ to $\stt[z_5]$), where sending has been done
successfully, if $\ag[S]$ does not know that $\ag[R]$ knows the bit; but it need
not make such progress, \ie, sending can be unsuccessful.  Similarly, the system
can only advance from $\stt[z_1]$ to $\stt[z_3]$ (and $\stt[z_5]$ to
$\stt[z_7]$), where an acknowledgement has been sent successfully, if $\ag[R]$
knows the bit and $\ag[R]$ does not know that $\ag[S]$ knows that $\ag[R]$ knows
the bit.


  \item\label{it:ex:egts:vs}
Consider the variable setting problem of the introduction for a single agent $\ag[a]$:
\begin{equation*}
\begin{array}[t]{@{}r@{\ }l@{\ }r@{}}
  \textbf{\texttt{if}} & \K{\ag[a]}{\prop[x] \neq 1} \rightarrowtriangle \prop[x] \gets 3
\\
  \talloblong          & \K{\ag[a]}{\prop[x] \neq 3} \rightarrowtriangle \prop[x] \gets 1 & \textbf{\texttt{fi}}
\end{array}
\end{equation*}
Encoding the integer $\prop[x] \in \{ 0, 1, 2, 3 \}$ by two bits $\prop[q_1]$
and $\prop[q_2]$, we model the problem as the following epistemically guarded
transition system $\Gamma_{\mathit{vs}} = (\ebas[vs],
\mathcal{T}_{\mathit{vs}})$ with $\ebas[vs] = (S_{\mathit{vs}},\allowbreak
E_{\mathit{vs}},\allowbreak L_{\mathit{vs}},\allowbreak S_{\mathit{vs}, 0})$
over $\esig[vs] = \ESig[vs]$ with $\Prop[vs] = \{ \prop[q_1], \prop[q_2] \}$ and
$\Ag[vs] = \{ \ag[a] \}$:
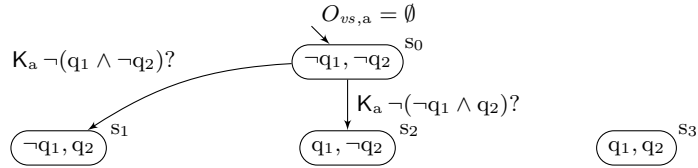
\begin{center}
\begin{tikzpicture}[scale=.95, transform shape, every loop/.style={-latex'}, every label/.style={label distance=0cm}, font={\fontsize{9pt}{9pt}\selectfont}], font={\fontsize{9pt}{9pt}\selectfont}]
\tikzstyle{state}=[rounded rectangle,draw,inner sep=4pt,minimum size=6pt]
\tikzstyle{transition}=[draw,-latex']
\tikzstyle{annotation}=[]
\node[state, label={[xshift=12pt, anchor=west]above right:$\stt[s_0]$}] (g0) at (0, 0) {$\neg\prop[q_1], \neg\prop[q_2]$};
\node[state, label={[xshift=10pt, anchor=west]above right:$\stt[s_1]$}] (g1) at (-4, -1.2) {$\neg\prop[q_1], \prop[q_2]$};
\node[state, label={[xshift=10pt, anchor=west]above right:$\stt[s_2]$}] (g2) at (0, -1.2) {$\prop[q_1], \neg\prop[q_2]$};
\node[state, label={[xshift=6pt, anchor=west]above right:$\stt[s_3]$}] (g3) at (4, -1.2) {$\prop[q_1], \prop[q_2]$};
%
\coordinate (W) at ($ (g0) + (-0.5, 0.5) $);
\node[annotation, inner sep=0pt, outer sep=0pt, anchor=south west, xshift=4pt] at (W) {$O_{\mathit{vs}, \ag[a]} = \emptyset$};
\path[transition]
  (W) -- (g0)
  (g0) edge[bend right=15] node[anchor=south east] {$\testact{\K{\ag[a]}{\neg(\prop[q_1] \land \neg\prop[q_2])}}$}
  (g1)
  (g0) edge node[anchor=west] {$\testact{\K{\ag[a]}{\neg(\neg\prop[q_1] \land \prop[q_2])}}$}
  (g2)
;
\end{tikzpicture}
\end{center}
$O_{\mathit{vs}, \ag[a]}$ represents a ``blind'' agent $\ag[a]$ that deems all
states equally accessible.  State $\stt[s_3]$ is definitely not
reachable. $\mathcal{T}_{\mathit{vs}}$ consists of two epistemically guarded
actions:
\begin{gather*}
\gact{\K{\ag[a]}{\neg(\prop[q_1] \land \neg\prop[q_2])}}{\{ (\stt[s_0], \stt[s_1]) \}}
\ \text{,}\\
\gact{\K{\ag[a]}{\neg(\neg\prop[q_1] \land \prop[q_2])}}{\{ (\stt[s_0],
  \stt[s_2]) \}}
\ \text{.}\qedhere
\end{gather*}
\end{refpars}
\end{example}

\subsection{Interpreting Epistemically Guarded Transition Systems}

An epistemically guarded transition system $\Gamma = (S, E, L, S_0,
\mathcal{T})$ over $\ESig$ is \emph{interpreted} over an epistemic transition
structure $M \in \ETS(S, E, L, S_0)$ by interpreting each guarded action
$(\gact{\varphi}{B}) \in \mathcal{T}$ \wrt to $M$ as
\begin{equation*}\textstyle
  \interp{(\gact{\varphi}{B})}{M} = \{ (s, s') \in B \mid s \in \RS(M) \text{ and } M, s \models \varphi \}
\ \text{,}
\end{equation*}
and combining these interpretations into the epistemic transition structure
\begin{equation*}\textstyle
  \interp{\Gamma}{M} = (S, E, L, S_0, \bigcup_{\tau \in \mathcal{T}} \interp{\tau}{M})
\ \text{.}
\end{equation*}
We call $M$ a \emph{solution} for $\Gamma$ if $\interp{\Gamma}{M} = M$.

\begin{example}\label{ex:interp-bt}
For the bit transmission problem as described in
\cref{ex:egts}\cref{it:ex:egts:bt}, the epistemic transition structure
$M_{\mathit{bt}} = (\ebas_{\mathit{bt}}, T_{\mathit{bt}})$ with $T_{\mathit{bt}}
= \{ (\stt[z_{\nof{i}}], \stt[z_{\nof{i}}]) \mid i \in \{ 0,\allowbreak 1,\allowbreak
3,\allowbreak 4,\allowbreak 5,\allowbreak 7 \} \} \cup \{ (\stt[z_0],
\stt[z_1]),\allowbreak (\stt[z_1], \stt[z_3]),\allowbreak (\stt[z_4],
\stt[z_5]),\allowbreak (\stt[z_5], \stt[z_7]) \}$ satisfies
$\interp{\Gamma_{\mathit{bt}}}{M_{\mathit{bt}}} = M_{\mathit{bt}}$.  This structure
just omits the states $\stt[z_2]$ and $\stt[z_6]$ with
$L_{\mathit{bt}}(\stt[z_2]) = \{ \prop[ack] \}$ and $L_{\mathit{bt}}(\stt[z_6])
= \{ \prop[sbit], \prop[ack] \}$ which are definitely not reachable, as
$\K{\ag[R]}{\mathit{sbit}}$ is false in $\stt[z_0] \sim_{O_{\mathit{bt},
    \ag[R]}} \stt[z_4]$.  Indeed,
\begin{gather*}
M_{\mathit{bt}}, s \models \neg\K{\ag[S]}{\K{\ag[R]}{\mathit{sbit}}}
\iff
s \in \{ \stt[z_0], \stt[z_1], \stt[z_4], \stt[z_5] \}
\\
M_{\mathit{bt}}, s \models \K{\ag[R]}{\mathit{sbit}}
\iff
s \in \{ \stt[z_1], \stt[z_3], \stt[z_5], \stt[z_7] \}
\\
M_{\mathit{bt}}, s \models \neg\K{\ag[R]}{\K{\ag[S]}{\K{\ag[R]}{\mathit{sbit}}}}
\iff
s \in \{ \stt[z_0], \stt[z_1], \stt[z_3], \stt[z_4], \stt[z_5], \stt[z_7] \}
\qedhere
\end{gather*}
\end{example}

However, finding a solution is complicated by the fact that the functional of
interpreting an epistemically guarded transition system over an epistemic
transition structure is not monotone, in general, as illustrated by the following
examples.

\begin{example}\label{ex:interp}
\begin{refpars}
  \item\label{it:ex:interp:vs} Continuing \cref{ex:egts}\cref{it:ex:egts:vs} for
the variable setting problem $\Gamma_{\mathit{vs}}$, consider the epistemic
transition structure $M_{\mathit{vs}, 0} \in \ETS[{\esig[vs]}](\ebas[vs])$ with
the empty transition relation $\Trans(M_{\mathit{vs}, 0}) = \emptyset$, and
hence $\RS[0](M_{\mathit{vs}, 0}) = \{ \stt[s_0] \}$.  Setting $M_{\mathit{vs},
  i+1} = \interp{\Gamma_{\mathit{vs}}}{M_{\mathit{vs}, i}}$ for $0 \leq i \leq
2$ we obtain successively
\begin{center}\fontsize{9.5pt}{11pt}\selectfont
\begin{tabular}{@{}c@{\quad}c@{\quad}c@{\quad}c@{}}\toprule
\multicolumn{1}{c}{$\tau$} & $\interp{\tau}{M_{\mathit{vs}, 0}}$ & $\interp{\tau}{M_{\mathit{vs}, 1}}$ & $\interp{\tau}{M_{\mathit{vs}, 2}}$
\\\midrule
$\gact{\K{\ag[a]}{\neg(\prop[q_1] \land \neg\prop[q_2])}}{\{ (\stt[s_0], \stt[s_1]) \}}$
&
$\{ (\stt[s_0], \stt[s_1]) \}$
&
$\emptyset$
&
$\{ (\stt[s_0], \stt[s_1]) \}$
\\
$\gact{\K{\ag[a]}{\neg(\neg\prop[q_1] \land \prop[q_2])}}{\{ (\stt[s_0], \stt[s_2]) \}}$
&
$\{ (\stt[s_0], \stt[s_2]) \}$
&
$\emptyset$
&
$\{ (\stt[s_0], \stt[s_2]) \}$
\\\bottomrule
\end{tabular}
\end{center}
In particular, $M_{\mathit{vs}, 2} =
\interp{\Gamma_{\mathit{vs}}}{M_{\mathit{vs}, 1}} =
\interp{\Gamma_{\mathit{vs}}}{\interp{\Gamma_{\mathit{vs}}}{M_{\mathit{vs}, 0}}}
= M_{\mathit{vs}, 0}$.  However, for $M_{\mathit{vs}, 4}, M_{\mathit{vs}, 5} \in
\ETS[{\esig[vs]}](\ebas[vs])$ with $\Trans(M_{\mathit{vs}, 4}) = \{
(\mathrm{s}_0, \mathrm{s}_1) \}$ and $\Trans(M_{\mathit{vs}, 5}) = \{
(\mathrm{s}_0, \mathrm{s}_2) \}$ we obtain that
$\interp{\Gamma_{\mathit{vs}}}{M_{\mathit{vs}, 4}} = M_{\mathit{vs}, 4}$ and
$\interp{\Gamma_{\mathit{vs}}}{M_{\mathit{vs}, 5}} = M_{\mathit{vs}, 5}$.

  \item\label{it:ex:interp:vsb} For capturing the cycle"=breaking variable
setting of the introduction consider the following epistemically guarded
transition system $\Gamma_{\mathit{vsb}} = (\ebas_{\mathit{vs}},
\mathcal{T}_{\mathit{vsb}})$ over $\esig[vs]$ that shares the epistemic state
basis $\ebas[vs]$ with \cref{ex:egts}\cref{it:ex:egts:vs}:
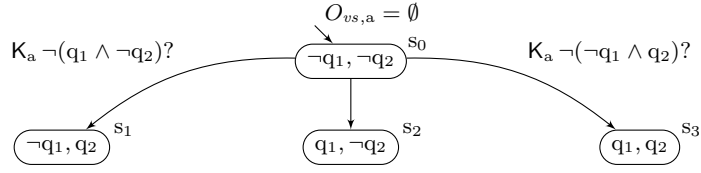
\begin{center}
\begin{tikzpicture}[scale=.95,transform shape,every loop/.style={-latex'}, every label/.style={label distance=0cm}, font={\fontsize{9pt}{9pt}\selectfont}], font={\fontsize{9pt}{9pt}\selectfont}]
\tikzstyle{state}=[rounded rectangle,draw,inner sep=4pt,minimum size=6pt]
\tikzstyle{transition}=[draw,-latex']
\tikzstyle{annotation}=[]
\node[state, label={[xshift=12pt, anchor=west]above right:$\stt[s_0]$}] (g0) at (0, 0) {$\neg\prop[q_1], \neg\prop[q_2]$};
\node[state, label={[xshift=10pt, anchor=west]above right:$\stt[s_1]$}] (g1) at (-4, -1.2) {$\neg\prop[q_1], \prop[q_2]$};
\node[state, label={[xshift=10pt, anchor=west]above right:$\stt[s_2]$}] (g2) at (0, -1.2) {$\prop[q_1], \neg\prop[q_2]$};
\node[state, label={[xshift=6pt, anchor=west]above right:$\stt[s_3]$}] (g3) at (4, -1.2) {$\prop[q_1], \prop[q_2]$};
\coordinate (W) at ($ (g0) + (-0.5, 0.5) $);
\node[annotation, inner sep=0pt, outer sep=0pt, anchor=south west, xshift=4pt] at (W) {$O_{\mathit{vs}, \ag[a]} = \emptyset$};
\path[transition]
  (W) -- (g0)
  (g0) edge[bend right=20] node[anchor=south east] {$\testact{\K{\ag[a]}{\neg(\prop[q_1] \land \neg\prop[q_2])}}$} (g1)
  (g0) edge (g2)
  (g0) edge[bend left=20] node[anchor=south west] {$\testact{\K{\ag[a]}{\neg(\neg\prop[q_1] \land \prop[q_2])}}$} (g3)
;
\end{tikzpicture}
\end{center}
For $M_{\mathit{vsb}, 0} = (\ebas[vs], \emptyset)$ with $\RS[0](M_{\mathit{vsb},
  0}) = \{ \stt[s_0] \}$, and setting $M_{\mathit{vsb}, i+1} =
\interp{\Gamma_{\mathit{vsb}}}{M_{\mathit{vsb}, i}}$ for $0 \leq i \leq 3$ we
obtain successively
\begin{center}\fontsize{9.5pt}{11pt}\selectfont
\begin{tabular}{@{}r@{\quad}c@{\quad}c@{\quad}c@{\quad}c@{}}\toprule
\multicolumn{1}{c}{$\tau$} & $\interp{\tau}{M_{\mathit{vsb}, 0}}$ & $\interp{\tau}{M_{\mathit{vsb}, 1}}$ & $\interp{\tau}{M_{\mathit{vsb}, 2}}$ & $\interp{\tau}{M_{\mathit{vsb}, 3}}$
\\\midrule
$\gact{\K{\ag[a]}{\neg(\prop[q_1] \land \neg\prop[q_2])}}{\{ (\stt[s_0], \stt[s_1]) \}}$
&
$\{ (\stt[s_0], \stt[s_1]) \}$
&
$\emptyset$
&
$\emptyset$
&
$\emptyset$
\\
$\gact{\truefrm}{\{ (\stt[s_0], \stt[s_2]) \}}$
&
$\{ (\stt[s_0], \stt[s_2]) \}$
&
$\{ (\stt[s_0], \stt[s_2]) \}$
&
$\{ (\stt[s_0], \stt[s_2]) \}$
&
$\{ (\stt[s_0], \stt[s_2]) \}$
\\
$\gact{\K{\ag[a]}{\neg(\neg\prop[q_1] \land \prop[q_2])}}{\{ (\stt[s_0], \stt[s_3]) \}}$
&
$\{ (\stt[s_0], \stt[s_3]) \}$
&
$\emptyset$
&
$\{ (\stt[s_0], \stt[s_3]) \}$
&
$\{ (\stt[s_0], \stt[s_3]) \}$
\\\bottomrule
\end{tabular}
\end{center}
For this epistemic transition structure
$M_{\mathit{vsb}, 3}$ with $\RS(M_{\mathit{vsb}, 3}) = \{ \stt[s_0], \stt[s_1],
\stt[s_3] \}$ it finally holds that
$\interp{\Gamma_{\mathit{vsb}}}{M_{\mathit{vsb}, 3}} = M_{\mathit{vsb}, 3}$.\qedhere
\end{refpars}
\end{example}


\subsection{Iteration Semantics}\label{sect:iteration}

For illustrating the non-monotonicity of the interpretation functional we have
started the interpretation sequence for $\Gamma$ with the smallest epistemic
transition structure which suggests to look for a smallest fixed point --- which
need not exist.  De Haan et al.~\cite{de-haan-et-al:fundinf:2004} argue that a
substitute consisting of the greatest fixed point would be more liberal.  They
construct a transfinite approximation sequence starting from an $N_0$ having all
states reachable.  For a successor ordinal $\alpha + 1$, the approximation
$N_{\alpha+1}$ is just the interpretation of $\Gamma$ in $N_{\alpha}$; for a
limit ordinal $\lambda$, the approximation $N_{\lambda}$ is ``the intersection
of unions of approximations that are sufficiently close to the
limit''~\cite[p.~7]{de-haan-et-al:fundinf:2004}.  The latter is preferred over a
union of intersections as it includes more states which implies less knowledge,
such that ``agents [know] facts only when there are good reasons for
them''~(ibid.).

More formally, define the following sequence of interpretations for an
epistemically guarded transition system $\Gamma = (S, E, L, S_0, \mathcal{T})$
over $\ESig$ by transfinite induction:
\begin{gather*}
N_0 = (S, E, L, S_0, S \times S)
\\
N_{\alpha + 1} = \interp{\Gamma}{N_{\alpha}}
\quad\text{for any successor ordinal $\alpha+1$}\\\textstyle
N_{\lambda} = \bigcap_{\alpha < \lambda} \bigcup_{\alpha \leq \beta < \lambda} N_{\beta}
\quad\text{for any limit ordinal $\lambda$}
\end{gather*}
Intersection and union refer to the transition relations of the epistemic
transition structures, cf.~\cref{sect:ets}.  Due to cardinality reasons some
epistemic transition structure has to occur repeatedly, that is, the ordinal
$\eta_{\Gamma} = \inf \{ \alpha \mid \text{ex.\ }\beta\text{ s.\,t.\ }\alpha <
\beta\text{ and }N_{\alpha} = N_{\beta} \}$ is well defined.  If $N_{\alpha+1}
\subseteq N_{\alpha}$ for all $\alpha \geq \eta_{\Gamma}$, then
$N_{\eta_{\Gamma}+1} = N_{\eta_{\Gamma}}$; otherwise there is some ordinal
$\alpha \geq \eta_{\Gamma}$ such that $N_{\alpha+1} \not\subseteq N_{\alpha}$.
Thus the ordinal $\alpha_{\Gamma} = \inf \{ \alpha \mid \eta_{\Gamma} \leq
\alpha \text{ and } (N_{\alpha} = N_{\alpha+1} \text{ or } N_{\alpha+1}
\not\subseteq N_{\alpha}) \}$ is well defined.  The \emph{iteration semantics}
of $\Gamma$ is defined as $N_{\alpha_{\Gamma}}$.

The iteration starts with all states reachable, leading to the greatest fixed
point if the interpretation functional is monotone.

\begin{example}\label{ex:iter}
\begin{refpars}
  \item\label{it:ex:iter:vs} For the variable setting problem
$\Gamma_{\mathit{vs}}$ of \cref{ex:egts}\cref{it:ex:egts:vs} the interpretation
sequence $(N_{\mathit{vs}, \alpha})_{0 \leq \alpha}$ starts with
$N_{\mathit{vs}, 0}$ showing $\Trans(N_{\mathit{vs}, 0}) = S_{\mathit{vs}}
\times S_{\mathit{vs}}$.  Using the epistemic transition structures from
\cref{ex:interp}\cref{it:ex:interp:vs} it holds that $N_{\mathit{vs}, k+1} =
\interp{\Gamma_{\mathit{vs}}}{N_{\mathit{vs}, k}} = M_{\mathit{vs}, 2}$ for $k$
even and $N_{\mathit{vs}, k+1} = M_{\mathit{vs}, 1}$ for $k \geq 1$ odd.  Thus,
$N_{\mathit{vs}, 1} = N_{\mathit{vs}, 3}$ such that $\eta_{\Gamma_{\mathit{vs}}}
= 1 = \alpha_{\Gamma_{\mathit{vs}}}$, since $\Trans(N_{\mathit{vs}, 2}) = \{
(\stt[s_0], \stt[s_0]), (\stt[s_0], \stt[s_1]), (\stt[s_0], \stt[s_2]) \}
\not\subseteq \emptyset = \Trans(N_{\mathit{vs}, 1})$.  Hence the iteration
semantics of $\Gamma_{\mathit{vs}}$ is given by $N_{\mathit{vs}, 1} =
M_{\mathit{vs}, 2}$; since its transition relation is empty,
$\Gamma_{\mathit{vs}}$ has the same iteration semantics as an epistemically
guarded transition system without any guarded actions.

  \item\label{it:ex:iter:vsb} Computing the iteration semantics sequence
$(N_{\mathit{vsb}, \alpha})_{0 \leq k}$ of the cycle"=breaking variable setting
$\Gamma_{\mathit{vsb}}$ of \cref{ex:interp}\cref{it:ex:interp:vsb} proceeds as
$N_{\mathit{vsb}, k} = M_{\mathit{vsb}, k+1}$.  Since this time the functional
is monotone from $\alpha = 1$ onwards, the iteration semantics is
$N_{\mathit{vsb}, 2}$.

  \item\label{it:ex:iter:nc} Consider the following epistemically guarded
transition system $\Gamma_{\mathit{nc}} = (\ebas_{\mathit{vs}},\allowbreak
\mathcal{T}_{\mathit{nc}})$ over $\esig_{\mathit{vs}}$ that shares the epistemic
basis $\ebas[vs]$ with the variable setting problem $\Gamma_{\mathit{vs}}$ of
\cref{it:ex:iter:vs} and only adds the guarded action
$\gact{\K{\ag[a]}{\neg\prop[q_2]}}{\{ (\stt[s_0], \stt[s_3]) \}}$:
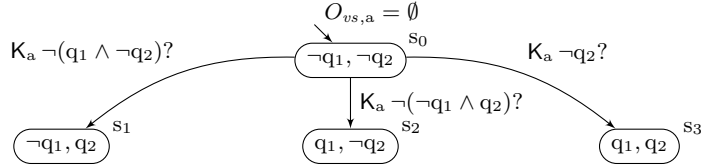
\begin{center}
\begin{tikzpicture}[scale=.95, transform shape, every loop/.style={-latex'}, every label/.style={label distance=0cm, inner sep=2pt, outer sep=0pt}, font={\fontsize{9pt}{9pt}\selectfont}], font={\fontsize{9pt}{9pt}\selectfont}]
\tikzstyle{state}=[rounded rectangle,draw,inner sep=4pt,minimum size=6pt]
\tikzstyle{transition}=[draw,-latex']
\tikzstyle{annotation}=[]
\node[state, label={[anchor=south west]10:$\stt[s_0]$}] (g0) at (0, 0) {$\neg\prop[q_1], \neg\prop[q_2]$};
\node[state, label={[anchor=south west]10:$\stt[s_1]$}] (g1) at (-4, -1.2) {$\neg\prop[q_1], \prop[q_2]$};
\node[state, label={[anchor=south west]10:$\stt[s_2]$}] (g2) at (0, -1.2) {$\prop[q_1], \neg\prop[q_2]$};
\node[state, label={[anchor=south west]10:$\stt[s_3]$}] (g3) at (4, -1.2) {$\prop[q_1], \prop[q_2]$};
\coordinate (W) at ($ (g0) + (-0.5, 0.5) $);
\node[annotation, inner sep=0pt, outer sep=0pt, anchor=south west, xshift=4pt] at (W) {$O_{\mathit{vs}, \ag[a]} = \emptyset$};
\path[transition]
  (W) -- (g0)
  (g0) edge[bend right=20] node[anchor=south east] {$\testact{\K{\ag[a]}{\neg(\prop[q_1] \land \neg\prop[q_2])}}$} (g1)
  (g0) edge node[anchor=west] {$\testact{\K{\ag[a]}{\neg(\neg\prop[q_1] \land \prop[q_2])}}$} (g2)
  (g0) edge[bend left=20] node[anchor=south west] {$\testact{\K{\ag[a]}{\neg\prop[q_2]}}$} (g3)
;
\end{tikzpicture}
\end{center}
The interpretation process runs as for $\Gamma_{\mathit{vs}}$, and the empty
transition relation is also the iteration semantics of $\Gamma_{\mathit{nc}}$.
There is, however, a unique non"=empty interpretation: the transition structure
consisting only of $(\stt[s_0], \stt[s_1])$.  Finding this solution, is not
constructive and some speculation is necessary: there is no solution with
$\stt[s_2]$ reachable, but not $\stt[s_1]$, since then the (non-)reachability of
$\stt[s_3]$ leads to a contradiction.  Thus only the possibility of $\stt[s_0]$
and $\stt[s_1]$ reachable, but $\stt[s_2]$ and $\stt[s_3]$ not reachable
remains.

  \item\label{it:ex:iter:maybe} For the epistemically guarded transition system
$\Gamma_{\mathit{may}}$ over $(\{ \prop[p] \}, \{ \ag[a] \})$ given by
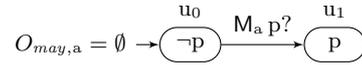
\begin{center}
\begin{tikzpicture}[scale=.95, transform shape, every loop/.style={-latex'}, every label/.style={label distance=0cm}, font={\fontsize{9pt}{9pt}\selectfont}], font={\fontsize{9pt}{9pt}\selectfont}]
\tikzstyle{state}=[rounded rectangle,draw,inner sep=4pt,minimum width=32pt]
\tikzstyle{transition}=[draw,-latex']
\tikzstyle{annotation}=[]
\node[state, label={[anchor=south]above:$\stt[u_0]$}] (g1) at (0, 0) {$\neg\prop[p]$};
\node[state, label={[anchor=south]above:$\stt[u_1]$}] (g2) at (2, 0) {$\prop[p]$};
\coordinate (W) at ($ (g1) + (-.75, 0) $);
\node[annotation, inner sep=0pt, outer sep=0pt, xshift=-4pt, anchor=east] at (W) {$O_{\mathit{may}, \ag[a]} = \emptyset$};
\path[transition]
  (W) -- (g1)
  (g1) edge node[anchor=south] {$\testact{\M{\ag[a]}{\prop[p]}}$}
  (g2)
;
\end{tikzpicture}
\end{center}
the iteration process when started with $N_{\mathit{may}, 0}$ having
$\Trans(N_{\mathit{may}, 0}) = \{ \stt[u_0], \stt[u_1] \} \times \{ \stt[u_0],
\stt[u_1] \}$ evaluates $\M{\ag[a]}{\prop[p]}$ to true and we obtain
$N_{\mathit{may}, 1}$ with $\Trans(N_{\mathit{may}, 1}) = \{ (\stt[u_0],
\stt[u_1]) \}$ which in turn is confirmed by the next iteration yielding a fixed
point.  This iteration semantics, however, has a touch of a ``vaticinium ex
eventu'': $\prop[p]$ can be reached since $\prop[p]$ may be reached.\qedhere
\end{refpars}
\end{example}

\subsection{Unique Interpretation Solutions}\label{sect:unique}

A knowledge"=based program can be executed reliably just step by step if each
knowledge guard can be stably decided based on what has been computed up to the
current point of execution.  In particular, in order to obtain a solution by
execution knowledge must not be invalidated by information only to be gained
later on.  Conversely, if all knowledge guards can be decided by just looking to
the past, there is at most a single solution.

Based on this observation, Fagin et
al.~\cite{fagin-et-al:dc:1997,fagin-et-al:2003} develop a formal
characterisation of unique interpretability by capturing the notion that
solutions ``depend on the past''.  They then show that ``providing epistemic
witnesses'' is a sufficient criterion for ``dependence on the past'', which in
turn always holds for ``synchronous'' programs.  We briefly summarise their main
line of argument adapting the demonstration from their run"=based account for
knowledge"=based programs to our state"=based epistemically guarded transition
systems.  The details of the adaptation are given in \cref{app:sect:kbp}.

An epistemic formula $\varphi \in \EFrm$ is said to \emph{depend on the past}
\wrt a class of epistemic transition structures $\mathcal{M} \subseteq \ETS(\ebas)$
if for all $M_1, M_2\in \mathcal{M}$ and all $k \in \NZ$ it holds that
$\RT[k](M_1) = \RT[k](M_2)$ implies $M_1, s \models \varphi \mathrel{\iff} M_2,
s \models \varphi$ for all $s \in \RS[k](M_1)\cap \RS[k](M_2)$; an epistemically
guarded transition system $\Gamma = (\ebas, \mathcal{T})$ over $\ESig$ is
\emph{depending on the past} \wrt $\mathcal{M}$ if every $\varphi$ in
$(\gact{\varphi}{B}) \in \mathcal{T}$ depends on the past \wrt $\mathcal{M}$.

\begin{example}
For \cref{ex:interp}\cref{it:ex:interp:vs} neither $\K{\ag[a]}{\neg(\prop[q_1
]  \land \neg\prop[q_2])}$ nor $\K{\ag[a]}{\neg(\neg\prop[q_1] \land
  \prop[q_2])}$ depends on the past \wrt $\{ M_{\mathit{vs}, 0}, M_{\mathit{vs},
  1} \}$.  In particular, $\RT[0](M_{\mathit{vs}, 0}) = \emptyset =
\RT[0](M_{\mathit{vs}, 1})$ and $\RS[0](M_{\mathit{vs}, 0}) = \{ \stt[s_0] \} =
\RS[0](M_{\mathit{vs}, 1})$, but $M_{\mathit{vs}, 0}, \stt[s_0] \models
\K{\ag[a]}{\neg(\prop[q_1] \land \neg\prop[q_2])}$ and $M_{\mathit{vs}, 1},
\stt[s_0] \not\models \K{\ag[a]}{\neg(\prop[q_1] \land \neg\prop[q_2])}$.
Similarly for \cref{ex:interp}\cref{it:ex:interp:vsb}, these two formulæ do not
depend on the past \wrt $\{ M_{\mathit{vsb}, 0}, M_{\mathit{vsb}, 1},
M_{\mathit{vsb}, 2}, M_{\mathit{vsb}, 3} \}$, but they do depend on the past
\wrt $\{ M_{\mathit{vsb}, 1}, M_{\mathit{vsb}, 2}, M_{\mathit{vsb}, 3} \}$.
\end{example}

An epistemically guarded transition system $\Gamma$ has at most one solution if,
and only if, it depends on the past \wrt all its solutions.  Due to the
dependence on the past the successive reachable transition relations $\RT[k](M)$
of all solutions $M = \interp{\Gamma}{M}$, \ie, their pasts, coincide.

\begin{restatable}{proposition}{uniqueness}\label{prop:unique}
Let $\Gamma = (\ebas, \mathcal{T})$ be an epistemically guarded transition
system over $\esig$.  Then $\Gamma$ has at most one solution if, and only if,
there is an $\mathcal{M} \subseteq \ETS[\esig](\ebas)$ with $\{ M \in
\ETS[\esig](\ebas) \mid \interp{\Gamma}{M} = M \} \subseteq \mathcal{M}$ such
that $\Gamma$ depends on the past \wrt $\mathcal{M}$.
\end{restatable}

In order to obtain a solution of $\Gamma$ by execution, the system is
interpreted repeatedly to construct the approximations $(M_k)_{0 \leq k}$ with
$M_{k+1} = \interp{\Gamma}{M_k}$ for $k \geq -1$ starting with some $M_{-1}$.
Each approximation $M_k$ with $k \geq 0$ contributes a transition relation
$\RT[k](M_k)$ which can be combined into a limit $M_{\omega}$.  If $\Gamma$
depends on the past \wrt the class of epistemic transition structures from which
the approximants are constructed and which also contains the limit, then the
interpretation of the limit $M_{\omega}$ yields a fixed point.

\begin{restatable}{proposition}{existence}\label{prop:ex}
Let $\Gamma = (\ebas, \mathcal{T})$ be an epistemically guarded transition
system over $\esig$.  Let $\mathcal{M} \subseteq \ETS[\esig](\ebas)$ such that
$\interp{\Gamma}{M} \in \mathcal{M}$ for every $M \in \mathcal{M}$ and $(\ebas,
\bigcup_{0 \leq k} \RT[k](M_k)) \in \mathcal{M}$ for all $(M_k)_{0 \leq k}
\subseteq \mathcal{M}$ with $\RT[k](M_{k'}) = \RT[k](M_k)$ for all $k' \geq k
\geq 0$.  Let $\Gamma$ depend on the past \wrt $\mathcal{M}$ and $M_{-1} \in
\mathcal{M}$, $M_{i+1} = \interp{\Gamma}{M_i}$ for all $i \geq -1$, and
$M_{\omega} = (\ebas, \bigcup_{0 \leq k} \RT[k](M_k))$.  Then
$\interp{\Gamma}{M_{\omega}} = \interp{\Gamma}{\interp{\Gamma}{M_{\omega}}}$.
\end{restatable}

A sufficient criterion for obtaining a comprehensive class of epistemic
transition structures $\mathcal{M}$ such that $\Gamma$ depends on the past \wrt
$\mathcal{M}$ is provided by epistemic witnesses: If some knowledge formula
$\K{\ag}{\varphi}$ of $\Gamma$ does not hold at some state of an interpreting
epistemic transition structure there is evidence in the past of this structure
why it does not hold.  Formally, a structure $M \in \ETS(\ebas)$ \emph{provides
  epistemic witnesses} for a formula $\K{\ag}{\varphi} \in \EFrm$ if for all $k
\geq 0$, $s \in \RS[k](M)$ it holds that if $M, s \not\models \K{\ag}{\varphi}$,
then there is an $s' \in \RS[k](M)$ with $(s, s') \in E_{\ag}$ and $M, s'
\not\models \varphi$.

\begin{restatable}{lemma}{providing}\label{prop:prov-wit}
Let $\Gamma = (\ebas, \mathcal{T})$ be an epistemically guarded
transition system over $\ESig$ and let $\mathcal{M} \subseteq \ETS(S, E, L,
S_0)$ such that all $M \in \mathcal{M}$ provide epistemic witnesses for all
knowledge guards in $\Gamma$.  Then $\Gamma$ is depending on the past \wrt
$\mathcal{M}$.
\end{restatable}

A sufficient criterion, in turn, for a structure $M \in \ETS(S,\allowbreak
E,\allowbreak L,\allowbreak S_0)$ to provide epistemic witnesses is $M$ being
\emph{synchronous}: if for all $\ag \in \Ag$ and all reachable $s_1 \in
\RS[k_1](M)$ and $s_2 \in \RS[k_2](M)$ with $(s_1, s_2) \in E_{\ag}$ it holds
that $s_1, s_2 \in \RS[\min\{ k_1, k_2 \}](M)$.  In a synchronous structure the
temporal and the epistemic dimension for each agent are hence tightly coupled
and no agent can access but also does not need to know the future.

\begin{example}
The interpretation $M_{\mathit{bt}}$ of the bit transmission problem given in
\cref{ex:interp-bt} provides epistemic witnesses, but is not synchronous: the
sender $\ag[S]$ cannot distinguish $\stt[z_0]$ reachable at depth $0$ of
$M_{\mathit{bt}}$ from $\stt[z_1]$ that is only reachable at depth $1$, and
similarly the receiver $\ag[R]$ cannot distinguish $\stt[z_1]$ from $\stt[z_3]$
at the respective depths of $1$ and $2$.
\end{example}

An epistemically guarded transition system $\Gamma = (\ebas, \mathcal{T})$ over
$\esig$ \emph{provides epistemic witnesses} if for each $M \in
\ETS[\esig](\ebas)$ the interpretation $\interp{\Gamma}{M}$ provides epistemic
witnesses for all knowledge formulæ occuring in some of the action guards of
$\Gamma$; $\Gamma$ is \emph{synchronous} if each $\interp{\Gamma}{M}$ is
synchronous.  Moreover, $\Gamma$ can syntactically be seen to be synchronous
(cf.~\cite[p.~135]{fagin-et-al:2003}) if it is round"=based where all agents
perform some action in each round and record locally which actions they have
taken.

\section{(Re-)Interpreting Knowledge"=based Programs}\label{sect:interp}

The results by Fagin et al.~\cite{fagin-et-al:dc:1997,fagin-et-al:2003}
guarantee a unique interpretation for all synchronous knowledge"=based programs;
the approach by De Haan et al.~\cite{de-haan-et-al:fundinf:2004} aims at
extending the interpretation to asynchronous programs, but assigns semantics
also to contradictory or self"=fulfilling programs.

The necessity of avoiding contradictory or self"=fulfilling behaviour already
occurs in the design of synchronous programming
languages~\cite{beneviste-et-al:ieee:2003}: Their underlying principle is
``perfect synchrony'', that any reaction of a program takes zero time and that
thus whatever is output in reaction to some input is already present at the same
time as the input.  Since the presence or absence of signals can be tested, this
requires ``logical coherence'' saying that a (non-input) signal is present in a
reaction if, and only if, this signal is emitted in this very reaction.  A
program needs to be both \emph{reactive} in the sense of leading to some
logically coherent signal status, and \emph{determinate}, \ie, not showing
several such statuses.  For example, in Esterel \cite{berry:milner:2000}, the
program fragment
\begin{lstlisting}[language=Esterel]
present S then nothing else emit S end
\end{lstlisting}
is not reactive, but contradictory: signal \code{S} is only emitted if it is not
emitted; and
\begin{lstlisting}[language=Esterel]
present S then emit S else nothing end
\end{lstlisting}
is not determinate, but self"=fulfilling: \code{S} is emitted if it is emitted,
and it is not emitted if it is not.  Such programs can be revealed by using a
cycle-detecting static analysis, as is done in
Lustre~\cite{halbwachs-et-al:ieee:1991}, or, for including more intricate cases,
by Berry's ``constructive semantics'' as for Esterel~\cite{berry:esterel}.
Building on a ``logical semantics'' recording what is emitted in each step of
execution, a \emph{must}/\emph{cannot} analysis is performed: what must/cannot
be emitted, which branch must/cannot be executed.  It is then required that for
each signal it can be decided whether it must be present or it cannot be
present.  For example, in the parallel execution
\begin{lstlisting}[language=Esterel]
   [ present S1 then emit S1 end ]
|| [ present S1 then present S2 then nothing else emit S2 end end ]
\end{lstlisting}
both signals can be emitted --- if \code{S1} is assumed to be present, and
\code{S2} absent ---, but none must be emitted.  Thus the constructive semantics
does not reach a decision of what must/cannot be present and the program is not
constructive.  Intriguingly, however, there is exactly one coherent signal
status that can be reached by execution: \code{S1} and \code{S2} absent.

We adapt Berry's constructive semantics approach to knowledge"=based programs.
In fact, the first, non"=reactive Esterel program fragment resembles the
variable setting problem described in \cref{ex:interp}\cref{it:ex:interp:vs},
the second, non"=determinate fragment directly corresponds to
\cref{ex:iter}\cref{it:ex:iter:maybe}, and the last, combined fragment is
essentially the same as \cref{ex:iter}\cref{it:ex:iter:nc}.  We first define a
must/can version of epistemic transition structures with a lower (must) and an
upper bound (can).  Based on a positive (must) and negative (cannot)
satisfaction relation of epistemic formulæ over these structures we show how an
epistemically guarded transition system can be interpreted yielding another
epistemic must/can transition structure.  For uniformity, we rephrase this
interpretation in terms of the negation normal form of formulæ and demonstrate
that the constructive interpretation is always monotone and leads to a least
fixed point.  For any knowledge"=based program, this fixed point soundly shows
which executions are necessary and which are possible.  However, the fixed point
need not be decided, and more can be possible than is necessary.  We show that
synchronous programs always lead to decided fixed points.  The proofs of all
steps are contained in \cref{app:sect:interp}.

\subsection{Epistemic Must/Can Transition Structures}

An \emph{epistemic must/can transition structure} $Y = (S, E, L, S_0, (T_{\mu},
T_{\nu}))$ over $\esig = \ESig$ is given by an epistemic state basis $\ebas =
(S, E, L, S_0)$ and two \emph{lower} and \emph{upper} \emph{transition
  relations} $T_{\mu}, T_{\nu} \subseteq S \times S$ with $T_{\mu} \subseteq
T_{\nu}$.  In particular, $Y_{\mu} = (\ebas, T_{\mu})$ and $Y_{\nu} =
(\ebas,\allowbreak T_{\nu})$ are epistemic transition structures over $\esig$
with $Y_{\mu} \subseteq Y_{\nu}$.

The \emph{positive} and \emph{negative satisfaction relations} of an epistemic
formula $\varphi \in \EFrm$ over the epistemic must/can transition structure $Y$
at a state $s \in \RS(Y_{\nu})$, written $Y, s \pmodels \varphi$ and $Y, s
\nmodels \varphi$, are defined as follows:
\begin{align*}
Y, s &\pmodels p \iff p \in L(s)
&
Y, s &\nmodels p \iff p \notin L(s)
\\
Y, s &\not\pmodels \falsefrm
&
Y, s &\nmodels \falsefrm
\\
Y, s &\pmodels \neg\varphi \iff Y, s \nmodels \varphi
&
Y, s &\nmodels \neg\varphi \iff Y, s \pmodels \varphi
\\
Y, s &\pmodels \varphi_1 \land \varphi_2 \iff{}
&
Y, s &\nmodels \varphi_1 \land \varphi_2 \iff{}
\\[-.3ex]
&\quad Y, s \pmodels \varphi_1 \text{ and } Y, s \pmodels \varphi_2
&
&\quad Y, s \nmodels \varphi_1 \text{ or } Y, s \nmodels \varphi_2
\\
Y, s &\pmodels \K{\ag}{\varphi} \iff Y, s' \pmodels \varphi
&
Y, s &\nmodels \K{\ag}{\varphi} \iff Y, s' \nmodels \varphi
\\[-.3ex]
&\quad\stacked{\text{for all } s' \in \RS(Y_{\nu})\\ \text{with } (s, s') \in E_{\ag}}
&
&\quad\stacked{\text{for some } s' \in \RS(Y_{\mu})\\ \text{with } (s, s') \in E_{\ag}}
\end{align*}
A formula is positively satisfied over $Y$ if it must be true given the upper
bound $Y_{\nu}$ of possible behaviour, it is negatively satisfied if it cannot
be true given the lower bound $Y_{\mu}$ of necessary behaviour.  In fact, it
holds that what must be true can also be true:

\begin{restatable}{lemma}{LEMposnegconstructive}\label{lem:pos-neg-epistemic}
Let $Y = (S, E, L, S_0, (T_{\mu}, T_{\nu}))$ be an epistemic must/can transition
structure over $\ESig$ and $\varphi \in \EFrm$.  Then for all $s \in
\RS(Y_{\nu})$, $Y, s \pmodels \varphi$ implies $Y, s \not\nmodels \varphi$.
\end{restatable}

The set of epistemic must/can transition structures over $\esig$ and the
epistemic state basis $\ebas$ is denoted by $\EMCTS[\esig](\ebas)$.  We say that
$Y_1 \sqsubseteq Y_2$ for $Y_1, Y_2 \in \EMCTS[\esig](\ebas)$ if $Y_{1, \mu}
\subseteq Y_{2, \mu}$ and $Y_{1, \nu} \supseteq Y_{2, \nu}$: an \emph{extension}
raises the lower bound and reduces the upper bound.

As with epistemic transition structures, an epistemically guarded transition system
$\Gamma = (S, E, L, S_0, \mathcal{T})$ over $\ESig$ can be interpreted over an
epistemic must/can transition structure $Y \in \EMCTS(S, E, L, S_0)$: The
\emph{interpretation} of a guarded action $(\gact{\varphi}{B}) \in \mathcal{T}$
\wrt to $Y$ is given by the pair $\interp{(\gact{\varphi}{B})}{Y} =
(\interp{(\gact{\varphi}{B})}{Y, \mu}, \interp{(\gact{\varphi}{B})}{Y, \nu})$
with
\begin{gather*}
  \interp{(\gact{\varphi}{B})}{Y, \mu} = \{ (s, s') \in B \mid s \in \RS(Y_{\mu}) \text{ and } Y, s \pmodels \varphi \}
\ \text{,}\\
  \interp{(\gact{\varphi}{B})}{Y, \nu} = \{ (s, s') \in B \mid s \in \RS(Y_{\nu}) \text{ and } Y, s \not\nmodels \varphi \}
\ \text{.}
\end{gather*}
By \cref{lem:pos-neg-epistemic} it holds that $\interp{\tau}{Y, \mu} \subseteq
\interp{\tau}{Y, \nu}$ for each $\tau \in \mathcal{T}$.  The \emph{constructive
  interpretation} of $\Gamma$ \wrt $Y$ is given by the epistemic must/can
transition structure
\begin{equation*}\textstyle
  \interp{\Gamma}{Y} = (S, E, L, S_0, (\bigcup_{\tau \in \mathcal{T}} \interp{\tau}{Y, \mu}, \bigcup_{\tau \in \mathcal{T}} \interp{\tau}{Y, \nu}))
\ \text{.}
\end{equation*}
This is well defined, \ie, $(\interp{\Gamma}{Y})_{\mu} \subseteq
(\interp{\Gamma}{Y})_{\nu}$.  We call $Y$ a \emph{constructive solution} for
$\Gamma$ if $\interp{\Gamma}{Y} = Y$; a constructive solution is \emph{decided}
if $Y_{\mu} = Y_{\nu}$.

Again as with epistemic transition structures, this interpretation over
epistemic must/can transition structures can be iterated for finally reaching a
stable structure --- and this time interpretation turns out to be monotone.

\begin{example}\label{ex:pn}
\begin{refpars}
  \item\label{it:ex:pn:vsb} Re-consider the cycle"=breaking variable setting
problem of \cref{ex:interp}\cref{it:ex:interp:vsb}.  We start the interpretation
in $Y_{\mathit{vsb}, 0} = (\ebas[vs], (\emptyset, S_{\mathit{vs}}^2))$ and
successively obtain the following epistemic must/can transition structures:
\begin{center}\fontsize{9.5pt}{11pt}\selectfont
\begin{tabular}{@{}r@{\quad}c@{\quad}c@{\quad}c@{\quad}c@{}}\toprule
\multicolumn{1}{c}{$\tau$} & $\interp{\tau}{Y_{\mathit{vsb}, 0}}$ & $\interp{\tau}{Y_{\mathit{vsb}, 1}}$ & $\interp{\tau}{Y_{\mathit{vsb}, 2}}$ & $\interp{\tau}{Y_{\mathit{vsb}, 3}}$
\\\midrule
\multirow{2}{*}{$\gact{\K{\ag[a]}{\neg(\prop[q_1] \land \neg\prop[q_2])}}{\{ (\stt[s_0], \stt[s_1]) \}}$}
&
$\emptyset$
&
$\emptyset$
&
$\emptyset$
&
$\emptyset$
\\
&
$\{ (\stt[s_0], \stt[s_1]) \}$
&
$\emptyset$
&
$\emptyset$
&
$\emptyset$
\\[1ex]
\multirow{2}{*}{$\gact{\truefrm}{\{ (\stt[s_0], \stt[s_2]) \}}$}
&
$\{ (\stt[s_0], \stt[s_2]) \}$
&
$\{ (\stt[s_0], \stt[s_2]) \}$
&
$\{ (\stt[s_0], \stt[s_2]) \}$
&
$\{ (\stt[s_0], \stt[s_2]) \}$
\\
&
$\{ (\stt[s_0], \stt[s_2]) \}$
&
$\{ (\stt[s_0], \stt[s_2]) \}$
&
$\{ (\stt[s_0], \stt[s_2]) \}$
&
$\{ (\stt[s_0], \stt[s_2]) \}$
\\[1ex]
\multirow{2}{*}{$\gact{\K{\ag[a]}{\neg(\neg\prop[q_1] \land \prop[q_2])}}{\{ (\stt[s_0], \stt[s_3]) \}}$}
&
$\emptyset$
&
$\emptyset$
&
$\{ (\stt[s_0], \stt[s_3]) \}$
&
$\{ (\stt[s_0], \stt[s_3]) \}$
\\
&
$\{ (\stt[s_0], \stt[s_3]) \}$
&
$\{ (\stt[s_0], \stt[s_3]) \}$
&
$\{ (\stt[s_0], \stt[s_3]) \}$
&
$\{ (\stt[s_0], \stt[s_3]) \}$
\\\bottomrule
\end{tabular}
\end{center}
Not only does it hold that $\interp{\Gamma_{\mathit{vsb}}}{Y_{\mathit{vsb}, 3}}
= Y_{\mathit{vsb}, 3}$, but the interpretations indeed evolve monotonically \wrt
$\sqsubseteq$.  Moreover, the structure $Y_{\mathit{vsb}, 3}$ is decided and
everything what can happen also must happen, \ie, $(Y_{\mathit{vsb}, 3})_{\mu} =
(Y_{\mathit{vsb}, 3})_{\nu}$.

  \item\label{it:ex:pn:vs} For the cyclic variable setting problem, see
\cref{ex:egts}\cref{it:ex:egts:vs} and \cref{ex:interp}\cref{it:ex:interp:vs},
the interpretation process is monotone, but only yields
\begin{center}\fontsize{9.5pt}{11pt}\selectfont
\begin{tabular}{@{}r@{\quad}c@{\quad}c@{}}\toprule
\multicolumn{1}{c}{$\tau$} & $\interp{\tau}{Y_{\mathit{vs}, 0}}$ & $\interp{\tau}{Y_{\mathit{vs}, 1}}$
\\\midrule
$\gact{\K{\ag[a]}{\neg(\prop[q_1] \land \neg\prop[q_2])}}{\{ (\stt[s_0], \stt[s_1]) \}}$
&
$(\emptyset, \{ (\stt[s_0], \stt[s_1]) \})$
&
$(\emptyset, \{ (\stt[s_0], \stt[s_1]) \})$
\\[.75ex]
$\gact{\K{\ag[a]}{\neg(\neg\prop[q_1] \land \prop[q_2])}}{\{ (\stt[s_0], \stt[s_2]) \}}$
&
$(\emptyset, \{ (\stt[s_0], \stt[s_2]) \})$
&
$(\emptyset, \{ (\stt[s_0], \stt[s_2]) \})$
\\\bottomrule
\end{tabular}
\end{center}
The epistemic must/can transition structure $Y_{\mathit{vs}, 1}$ is not decided,
and indeed there are two solutions of $\Gamma_{\mathit{vs}}$ in terms of
epistemic transition structures.  However, the same undecidedness holds true for
the non-constructive $\Gamma_{\mathit{nc}}$ of
\cref{ex:iter}\cref{it:ex:iter:nc}, that is, the unique solution is also missed
by the constructive interpretation.\qedhere
\end{refpars}
\end{example}

\subsection{Constructive Interpretation}\label{sect:constr}

The separated positive (must) and negative (cannot) satisfaction relations over
an epistemic must/can transition structure $Y \in \EMCTS(S, E, L, S_0)$ can be
merged into a single, uniform satisfaction relation relying on the
\emph{negation normal form} of epistemic formulæ where negation only occurs in
front of propositions.  For an arbitrary $\varphi \in \EFrm$ there exists an
equivalent $\nnf(\varphi) \in \EFrm$ in negation normal form, such that, in
particular
\begin{align*}&
  \nnf(\neg \prop) = \neg \prop
&&
  \nnf(\neg\neg\varphi) = \nnf(\varphi)
\\&
  \nnf(\neg\falsefrm) = \truefrm
&&
  \nnf(\neg(\varphi_1 \land \varphi_2)) = \nnf(\neg\varphi_1) \lor \nnf(\neg\varphi_2)
\\&
&&
  \nnf(\neg\K{\ag}{\varphi}) = \M{\ag}{\nnf(\neg\varphi)}
\end{align*}
The \emph{constructive satisfaction relation} $Y, s \cmodels \varphi$ for a
state $s \in \RS(Y_{\nu})$ and an epistemic formula $\varphi \in \EFrm$ in
negation normal form is defined just as for arbitrary epistemic formulæ, but
using the upper bound $Y_{\nu}$ for the universal quantifier of $\K{\ag}{}$ and
the lower bound $Y_{\mu}$ for the existential quantifier of $\M{\ag}{}$; in
particular,
\begin{gather*}
  Y, s \cmodels \neg p \iff p \notin L(s)
\\
  Y, s \cmodels \K{\ag}{\varphi} \iff \text{$Y, s' \cmodels \varphi$ f.\,a.\ $s' \in \RS(Y_{\nu})$ with $(s, s') \in E_{\ag}$}
\\
  Y, s \cmodels \M{\ag}{\varphi} \iff \text{ex.\ $s' \in \RS(Y_{\mu})$ s.\,t.\ $(s, s') \in E_{\ag}$ and $Y, s' \cmodels \varphi$}
\end{gather*}
The constructive satisfaction relation indeed combines $\pmodels$ and
$\nmodels$:

\begin{restatable}{lemma}{LEMposnegnnf}\label{lem:pos-neg-constructive}
Let $Y \in \EMCTS(S, E, L, S_0)$, $\varphi \in \EFrm$, and $s \in \RS(Y_{\nu})$.
Then $Y, s \pmodels \varphi$ iff $Y, s \cmodels \nnf(\varphi)$ and $Y, s
\nmodels \varphi$ iff $Y, s \cmodels \nnf(\neg\varphi)$.
\end{restatable}

It follows that if $Y_{\mu} = Y_{\nu}$, then $Y, s \models \varphi$ if, and only
if, $Y_{\mu}, s \models \varphi$ or, equivalently, $Y_{\nu}, s \models \varphi$.
We also obtain that constructive satisfaction is preserved when extending
epistemic must/can transition structures:

\begin{restatable}{lemma}{LEMextepistemic}\label{lem:ext-epistemic}
Let $Y, Y' \in \EMCTS(S, E, L, S_0)$ with $Y \sqsubseteq Y'$ and let $\varphi
\in \EFrm$.  Then $Y, s \cmodels \nnf(\varphi)$ implies $Y', s \cmodels
\nnf(\varphi)$ for all $s \in \RS(Y'_{\nu})$.
\end{restatable}


This preservation of satisfaction yields that constructive interpretation is
monotone.

\begin{restatable}{proposition}{PROPconstructivemonotone}\label{prop:constr-monotone}
Let $\Gamma = (S, E, L, S_0, \mathcal{T})$ be an epistemically guarded
transition system over $\ESig$ and $Y, Y' \in \EMCTS(S,\allowbreak E, L, S_0)$
such that $Y \sqsubseteq Y'$.  Then $\interp{\Gamma}{Y} \sqsubseteq
\interp{\Gamma}{Y'}$.
\end{restatable}

Finally, we can observe that $\EMCTS[\esig](\ebas)$ for $\ebas = (S, E, L, S_0)$
with the ordering $\sqsubseteq$ is an \emph{inductive partial order}: each
directed subset $\Delta \subseteq \EMCTS[\esig](\ebas)$ has a least upper bound
$\bigsqcup \Delta$ \wrt $\sqsubseteq$, where \emph{directed} means that every
two $Y_1, Y_2 \in \Delta$ have an upper bound $Y \in \Delta$ such that $Y_1
\sqsubseteq Y$ and $Y_2 \sqsubseteq Y$; and there is also a \emph{bottom} or
least element $\bot_{\esig, \ebas} = (S, E, L, S_0, (\emptyset, S \times S)) \in
\EMCTS[\esig](\ebas)$.

\begin{restatable}{proposition}{PROPconstructiveinductive}\label{prop:constr-inductive}
$(\EMCTS[\esig](\ebas), {\sqsubseteq}, \bot_{\esig, \ebas})$ is an inductive
partial order.
\end{restatable}

Pataraia's fixed"=point theorem~\cite[§8.22]{davey-priestley:2002} now
guarantees that the monotone operator $Y \mapsto \interp{\Gamma}{Y}$ for each
epistemically guarded transition system $\Gamma = (\ebas, \mathcal{T})$ has a
least fixed point in the inductive partial order.  It can be computed by,
possibly transfinite, iterated application of constructive interpretation to
$\bot_{\esig, \ebas}$, that is, $Y_0 = \bot_{\esig, \ebas}$, $Y_{\alpha+1} =
\interp{\Gamma}{Y_{\alpha}}$ for a successor ordinal $\alpha+1$, and
$Y_{\lambda} = \bigsqcup_{\alpha < \lambda} Y_{\alpha}$ until
equality~\cite[Exc.~8.19]{davey-priestley:2002}.  Compared to the iteration
semantics of \cref{sect:iteration}, the computation of the constructive
semantics thus does not have to record all previous approximations in order to
find a repetition.



\subsection{(Un-)Decided Constructive Fixed Points}\label{sect:decided}

If any constructive fixed point $Y = \interp{\Gamma}{Y}$ with $Y \in
\EMCTS[\esig](\ebas)$ is decided, then there is the solution $Y_{\mu} =
\interp{\Gamma}{Y_{\mu}} = \interp{\Gamma}{Y_{\nu}} = Y_{\nu}$ in terms of
epistemic transition structures, and $\Gamma$ is not contradictory.  Even if it
is not decided, the must/can structures $Y_{\mu\mu} = (\ebas,\allowbreak
(\Trans(Y_{\mu}), \Trans(Y_{\mu}))) \in \EMCTS[\esig](\ebas)$ and $Y_{\nu\nu} =
(\ebas,\allowbreak (\Trans(Y_{\nu}), \Trans(Y_{\nu}))) \in \EMCTS[\esig](\ebas)$
satisfy $Y \sqsubseteq Y_{\mu\mu}$ and $Y \sqsubseteq Y_{\nu\nu}$, such that by
\cref{prop:constr-monotone} we obtain $Y = \interp{\Gamma}{Y} \sqsubseteq
\interp{\Gamma}{Y_{\mu\mu}}, \interp{\Gamma}{Y_{\nu\nu}}$ which yields
%
$Y_{\mu} \subseteq \interp{\Gamma}{Y_{\mu}}$ and
$\interp{\Gamma}{Y_{\nu}} \subseteq Y_{\nu}$,
%
but not equality, in general.  For the least constructive fixed point
$\mu\Gamma$, any solution $M = \interp{\Gamma}{M}$ thus satisfies
$(\mu\Gamma)_{\mu} \subseteq M \subseteq (\mu\Gamma)_{\nu}$, always giving sound
lower and upper bounds and, if $\mu\Gamma$ is decided, moreover unique
solvability:


\begin{restatable}{proposition}{PROPconstructiveunique}\label{prop:constr-unique}
Let $\Gamma = (\ebas, \mathcal{T})$ be an epistemically guarded transition
system over $\esig$ and let $\mu\Gamma \in \EMCTS[\esig](\ebas)$ be decided.
Then $\Gamma$ has a unique solution in $\ETS[\esig](\ebas)$.
\end{restatable}

Still, even for epistemically guarded transition systems that
provide epistemic witnesses it is not guaranteed that the least constructive
fixed point is decided:

\begin{example}\label{ex:undecided}
Consider the following epistemically guarded transition system
$\Gamma_{\mathit{nd}} = (\ebas[\mathit{nd}],\allowbreak
\mathcal{T}_{\mathit{nd}})$ over $\esig[nd] = (\Prop[nd], \Ag[nd])$ with
$\Prop[nd] = \{ \prop[p], \prop[q] \}$ and $\Ag[nd] = \{ \ag[a], \ag[b] \}$:
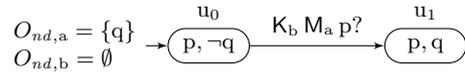
\begin{center}\vspace*{-.15ex}
\begin{tikzpicture}[scale=.95, transform shape, every loop/.style={-latex'}, every label/.style={label distance=0cm}, font={\fontsize{9pt}{9pt}\selectfont}], font={\fontsize{9pt}{9pt}\selectfont}]
\tikzstyle{state}=[rounded rectangle,draw,inner sep=4pt,minimum width=40pt]
\tikzstyle{transition}=[draw,-latex']
\tikzstyle{annotation}=[]
\node[state, label={[anchor=south]above:$\stt[u_0]$}] (g1) at (0, 0) {$\prop[p], \neg\prop[q]$};
\node[state, label={[anchor=south]above:$\stt[u_1]$}] (g2) at (3, 0) {$\prop[p], \prop[q]$};
\coordinate (W) at ($ (g1.west) + (-.3, 0) $);
\node[annotation, inner sep=0pt, outer sep=0pt, xshift=-4pt, align=left, anchor=east] at (W) {$O_{\mathit{nd}, \ag[a]} = \{ \prop[q] \}$\\[.5ex] $O_{\mathit{nd}, \ag[b]} = \emptyset$};
\path[transition]
  (W) -- (g1)
  (g1) edge node[anchor=south] {$\testact{\K{\ag[b]}{\M{\ag[a]}{\prop[p]}}}$}
  (g2)
;
\end{tikzpicture}\vspace*{-.15ex}
\end{center}
Constructive interpretation yields the non"=decided fixed point
$Y_{\mathit{nd}}$ with $\Trans(Y_{\mathit{nd}, \mu}) = \emptyset$ and
$\Trans(Y_{\mathit{nd}, \nu}) = \{ (\stt[u_0], \stt[u_1]) \}$, as
$Y_{\mathit{nd}}, \stt[u_0] \not\cmodels \K{\ag[b]}{\M{\ag[a]}{\prop[p]}}$, but
also $Y_{\mathit{nd}}, \stt[u_0] \not\cmodels
\M{\ag[b]}{\K{\ag[a]}{\neg\prop[p]}}$: the states $\stt[u_0]$ and $\stt[u_1]$
can be distinguished by agent $\ag[a]$, and agent $\ag[b]$ cannot tell whether a
step has been taken.  In $\stt[u_0]$ the formula $\M{\ag[a]}{\prop[p]}$ holds
\wrt $Y_{\mathit{nd}}$, but $\stt[u_1]$, does not, since
$(\stt[u_1],\stt[u_0])\not\in E_{\mathit{nd},a}$.  On the other hand,
$\Gamma_{\mathit{nd}}$ provides epistemic witnesses pathologically, since
$\interp{\Gamma_{\mathit{nd}}}{M}, s \models \K{\ag[b]}{\M{\ag[a]}{\prop[p]}}$
for any $M \in \ETS[{\esig[nd]}](\ebas[nd])$ and any $s \in
\RS(\interp{\Gamma_{\mathit{nd}}}{M})$, and hence has a unique interpretation,
which in this case is $\interp{\Gamma_{\mathit{nd}}}{Y_{\mathit{nd}, \mu}} =
Y_{\mathit{nd}, \nu} = \interp{\Gamma_{\mathit{nd}}}{Y_{\mathit{nd}, \nu}}$.
\end{example}

For synchronous epistemically guarded transition systems, however, the least
fixed point is decided, since all knowledge refers to a past that must have
happened:

\begin{restatable}{lemma}{LEMconstructivesync}\label{lem:constr-sync}
Let $\Gamma = (\ebas, \mathcal{T})$ be an epistemically guarded transition
system over $(P, A)$ that is synchronous.  Let $Y \in \EMCTS(\ebas)$ satisfy
$\interp{\Gamma}{Y} = Y$.  Then $Y$ is decided.
\end{restatable}

Summing up, the constructive approach to interpreting knowledge"=based programs
subsumes the solutions for synchronous programs and provides a sound procedure
for obtaining lower and upper bounds for the execution of both synchronous and
asynchronous knowledge"=based programs.  The approach, however, is not complete:
If the least constructive fixed point $\mu\Gamma$ is undecided, a system
$\Gamma$ may be contradictory without any solution (see
\cref{ex:interp}\cref{it:ex:interp:vs}), self"=fulfilling with several solutions
(see \cref{ex:iter}\cref{it:ex:iter:maybe}), or it may have a unique solution in
terms of epistemic transition structures (see
\cref{ex:iter}\cref{it:ex:iter:nc}).  One strategy that suggests itself for
analysing $\Gamma$ further is to check whether an interpretation using the lower
bound $(\mu\Gamma)_{\mu}$ of the least fixed point satisfies
$\interp{\Gamma}{(\mu\Gamma)_{\mu}} = (\mu\Gamma)_{\nu} =
\interp{\Gamma}{(\mu\Gamma)_{\nu}}$, which means that when executing according
to what must happen all what can happen is already covered (see
\cref{ex:undecided}).  Another option --- which, however, would assign a meaning
to any knowledge"=based program --- is to re"=interpret each guarded action
$\gact{\varphi}{B}$ as contributing $\{ (s, s') \mid s \in
\RS((\mu\Gamma)_{\nu}) \text{ and } \mu\Gamma, s \cmodels \nnf(\varphi) \}$ to
the overall semantics.

\section{Knowledge"=based Programs as Rule Systems}\label{sect:rules}

The ``executions'' of an epistemically guarded transition system $\Gamma$ can be
captured as derivations of two mutually dependent inductive rule systems, like
used for inductive definitions~\cite{aczel:hml:1977,harper:2013}.  One rule
system defines the reachability in $\Gamma$, the other one the satisfaction of
knowledge formulæ in negation normal form over $\Gamma$.  When $\Gamma$ provides
epistemic witnesses, the mutual dependence can be resolved by stratifying the
rule system for reachability according to the depth of the execution.  In the
general case, the non-monotone dependence of the formula satisfaction system on
the reachability system --- the more states are reachable, the less is known ---
can be mitigated by extending the notion of rule systems to include also
negative premisses: The conclusion of a rule is derivable if all its (positive)
premisses are derivable, but none of its negative premisses.  When applied to
knowledge formulæ, the negative premisses express that no counterexample can be
reachable.

\subsection{Inductive Rule Systems}

An \emph{inductive rule system} $R$ consists of \emph{rules} of the form $X/y$
where the \emph{premisses} $X \subseteq U$ and the \emph{conclusion} $y \in U$
are drawn from some \emph{universe} of \emph{judgements} $U$; a rule $X/y$ is
interpreted as ``if all $X$ can be inferred, then $y$ can be inferred''.  The
\emph{derivations} in $R$ together with their \emph{sets of premisses} and
\emph{conclusions} are inductively defined as follows:
\begin{itemize}
  \item a $y \in U$ is itself a derivation; its set of premisses is $\{ y \}$,
its conclusion is $y$;

  \item if $X/y \in R$ and $(d_x)_{x \in X}$ a family of derivations with
conclusions $(x)_{x \in X}$, then $(d_x)_{x \in X}/y$ is a derivation; its set
of premisses is the union of the premisses of $(d_x)_{x \in X}$, its conclusion
is $y$.
\end{itemize}
A $y \in U$ is \emph{derivable} in $R$ if there is a derivation in $R$ with the
empty set of premisses and conclusion $y$.  The set of derivable conclusions of
$R$ coincides with the least fixed point $\mu \hat{R}$ of $\hat{R} : \powerset U
\to \powerset U$ defined by $\hat{R}(P) = \{ y \in U \mid \text{ex.\ } X/y \in R
\text{ s.\,t.\ } X \subseteq P \}$.

For expressing reachability and the satisfaction of knowledge formulæ in an
epistemically guarded transition system $\Gamma = (S, E, L, S_0, \mathcal{T})$
over $\ESig$ as inductive rule systems, we use two types of judgements, one of
the form $s \in^{\Gamma} \RS$ with $s \in S$ for ``state $s$ is reachable in
$\Gamma$'', and one of the form $s \models^{\Gamma} \varphi$ with $s \in S$ and
$\varphi \in \EFrm$ in negation normal form for ``state $s$ satisfies formula
$\varphi$ in $\Gamma$''.
The rules for reachability read:
\begin{gather*}
  \structrule{}{s_0 \in^{\Gamma} \RS}
\quad\text{if $s_0 \in S_0$}
\qquad
  \structrule{s \in^{\Gamma} \RS}{s' \in^{\Gamma} \RS}
\quad\stackedtext{if ex.\ $(\gact{\varphi}{B}) \in \mathcal{T}$,\\\phantom{if }$(s, s') \in B$, and $s \models^{\Gamma} \varphi$}
\end{gather*}
where $s \models^{\Gamma} \varphi$ in the side condition of the second rule
requires this judgement to be derivable in the rule system for satisfaction.
The rules for this system read:
\begin{gather*}
  \structrule{}{s \models^{\Gamma} \truefrm}
\quad\stackedtext{if $s \in^{\Gamma} \RS$}
\qquad
  \structrule{}{s \models^{\Gamma} p}
\quad\stackedtext{if $s \in^{\Gamma} \RS$,\\\phantom{if }$p \in L(s)$}
\qquad
  \structrule{}{s \models^{\Gamma} \neg p}
\quad\stackedtext{if $s \in^{\Gamma} \RS$,\\\phantom{if }$p \notin L(s)$}
\\
  \structrule{s \models^{\Gamma} \varphi_1\quad s \models^{\Gamma} \varphi_2}{s \models^{\Gamma} \varphi_1 \land \varphi_2}
\qquad
  \structrule{s \models^{\Gamma} \varphi_1}{s \models^{\Gamma} \varphi_1 \lor \varphi_2}
\qquad
  \structrule{s \models^{\Gamma} \varphi_2}{s \models^{\Gamma} \varphi_1 \lor \varphi_2}
\\
  \structrule{s' \models^{\Gamma} \varphi}{s \models^{\Gamma} \M{a}{\varphi}}
\quad\stackedtext{if $(s, s') \in E_a$,\\\phantom{if }$s' \in^{\Gamma} \RS$}
\qquad
  \structrule{(s' \models_{\Gamma} \varphi)_{s' \in^{\Gamma} \RS,\ (s, s') \in E_a}}{s \models^{\Gamma} \K{a}{\varphi}}
\end{gather*}
Here, the last rule for satisfaction in fact is not monotone \wrt reachability:
In order to infer $s \models^{\Gamma} \K{a}{\varphi}$ it is not necessary to
infer $s' \models^{\Gamma} \varphi$ for all $s'$ with $(s, s') \in E_a$, but
only for those for which $s' \in^{\Gamma} \RS$ can be deduced --- and also for
all of those.

The notion of providing epistemic witnesses allows to stratify the inductive
rule systems according to the involved depth $k \geq 0$: We specialise the
judgement $s \in^{\Gamma} \RS$ into $s \in^{\Gamma} \RS[k]$ meaning ``state $s$
is reachable in $\Gamma$ in up to $k$ steps'' and, similarly, the judgement $s
\models^{\Gamma} \varphi$ into $s \models^{\Gamma}_k \varphi$ meaning ``formula
$\varphi$ is satisfied in $\Gamma$ at state $s$ considering states reachable in
up to $k$ steps''.  The rules for reachability become for all $k \geq 0$:
\begin{gather*}
  \structrule{}{s_0 \in^{\Gamma} \RS[k]}
\quad\text{if $s_0 \in^{\Gamma} S_0$}
\qquad
  \structrule{s \in^{\Gamma}  \RS[k]}{s' \in^{\Gamma} \RS[k+1]}
\quad\stackedtext{if ex.\ $(\gact{\varphi}{B}) \in \mathcal{T}$,\\\phantom{if }$(s, s') \in B$, and $s \models^{\Gamma}_k \varphi$}
\end{gather*}
Analogously the rules for satisfaction become for all $k \geq 0$:
\begin{gather*}
  \structrule{}{s \models^{\Gamma}_k \truefrm}
\quad\stackedtext{if $s \in^{\Gamma} \RS[k]$}
\qquad
  \structrule{}{s \models^{\Gamma}_k p}
\quad\stackedtext{if $s \in^{\Gamma} \RS[k]$,\\\phantom{if }$p \in L(s)$}
\qquad
  \structrule{}{s \models^{\Gamma}_k \neg p}
\quad\stackedtext{if $s \in^{\Gamma} \RS[k]$,\\\phantom{if }$p \notin L(s)$}
\\
  \structrule{s \models^{\Gamma}_k \varphi_1\quad s \models^{\Gamma}_k \varphi_2}{s \models^{\Gamma}_k \varphi_1 \land \varphi_2}
\qquad
  \structrule{s \models^{\Gamma}_k \varphi_1}{s \models^{\Gamma}_k \varphi_1 \lor \varphi_2}
\qquad
  \structrule{s \models^{\Gamma}_k \varphi_2}{s \models^{\Gamma}_k \varphi_1 \lor \varphi_2}
\\
  \structrule{s' \models^{\Gamma}_k \varphi}{s \models^{\Gamma}_k \M{a}{\varphi}}
\quad\stackedtext{if $(s, s') \in E_a$,\\\phantom{if }$s' \in^{\Gamma} \RS[k]$}
\qquad
  \structrule{(s' \models^{\Gamma}_k \varphi)_{s' \in^{\Gamma} \RS[k],\ (s, s') \in E_a}}{s \models^{\Gamma}_k \K{a}{\varphi}}
\end{gather*}
In particular, the rules for $s \models^{\Gamma}_k \M{a}{\varphi}$ and $s
\models^{\Gamma}_k \K{a}{\varphi}$ are sound for epistemically guarded
transition systems providing epistemic witnesses.

The notion of ``providing epistemic witnesses'' requires that, if
$\K{\ag}{\varphi}$ does not hold at reachability depth $k$, there is a
counterexample to $\varphi$ at reachability depth $\leq k$.  The general case
can be covered by dropping the depth requirement and taking into account that
$\K{\ag}{\varphi}$ does not hold at some state $s$ if, and only if, there is
some reachable, $\ag$-indistinguishable state $s'$ at which $\varphi$ does not
hold.  Therefore, in order to derive that $\K{\ag}{\varphi}$ indeed holds at
some reachable state $s$, it is necessary and sufficient to show that it is
\emph{not} possible to derive that $\neg\varphi$ holds at some reachable,
$\ag$-indistinguishable state $s'$.

\subsection{General Rule Systems with Positive and Negative Premisses}

For expressing negative information in terms of a rule system, we complement the
positive premisses of the rules by negative ones:
We consider general \emph{rule systems} $R$ over a universe $U$ consisting of
rules of the form $(X, \negat Z)/y$ where $X, Z \subseteq U$ are the
\emph{positive} and \emph{negative premisses}, and $y \in U$ is the
\emph{conclusion}; it is interpreted as ``if all $X$ can be inferred but no $Z$,
then $y$ can be inferred''.
The \emph{derivations} in $R$ together with their \emph{sets of positive and
  negative premisses} and \emph{conclusions} are again inductively defined as
follows:
\begin{itemize}
  \item a $y \in U$ is itself a derivation; its set of positive premisses is $\{
y \}$, its set of negative premisses is $\emptyset$, and its conclusion is $y$;

  \item if $(X, \negat Z)/y \in R$ and $(d_x)_{x \in X}$ a family of derivations
with conclusions $(x)_{x \in X}$, then $((d_x)_{x \in X}, \negat Z)/y$ is a
derivation; its set of positive premisses is the union of the positive premisses
of $(d_x)_{x \in X}$, its set of negative premisses is the union of the negative
premisses of $(d_x)_{x \in X}$ together with $Z$, and its conclusion is $y$.
\end{itemize}
For a $B \subseteq U$, let $\bar{R}(B)$ be all those $y \in U$ such that there
is a derivation of $y$ in $R$ with the empty set of positive premisses and no
negative premisses in $B$.  The set of \emph{derivable conclusions} of $R$ is
given by the least fixed point of $\bar{R}$ if it exists.

With this generalised notion of rule systems we can reformulate and combine the
two inference systems for reachability and satisfaction in an epistemically
guarded transition system $\Gamma = (S, E, L, S_0, \mathcal{T})$ over $\ESig$ by
using a single
judgement $s \models^{\Gamma}_{\omega} \varphi$ for ``state $s$ satisfies
$\varphi$ in $\Gamma$ and state $s$ is reachable in $\Gamma$''.  A negative
premiss $\negat(s \models^{\Gamma}_{\omega} \truefrm)$ thus stands for ``$s
\in^{\Gamma} \RS$ cannot be deduced''.  The new rules with also negative
premisses read:
\begin{gather*}
  \structrule{}{s_0 \models^{\Gamma}_{\omega} \truefrm}
\quad\text{if $s_0 \in S_0$}
\qquad
  \structrule{s \models^{\Gamma}_{\omega} \varphi}{s' \models^{\Gamma}_{\omega} \truefrm}
\quad\stackedtext{if ex.\ $(\gact{\varphi}{B}) \in \mathcal{T}$,\\\phantom{if }$(s, s') \in B$}
\\
  \structrule{s \models^{\Gamma}_{\omega} \truefrm}{s \models^{\Gamma}_{\omega} p}
\quad\text{if $p \in L(s)$}
\qquad
  \structrule{s \models^{\Gamma}_{\omega} \truefrm}{s \models^{\Gamma}_{\omega} \neg p}
\quad\text{if $p \notin L(s)$}
\\
  \structrule{s \models^{\Gamma}_{\omega} \varphi_1\quad s \models^{\Gamma}_{\omega} \varphi_2}{s \models^{\Gamma}_{\omega} \varphi_1 \land \varphi_2}
\qquad
  \structrule{s \models^{\Gamma}_{\omega} \varphi_1}{s \models^{\Gamma}_{\omega} \varphi_1 \lor \varphi_2}
\qquad
  \structrule{s \models^{\Gamma}_{\omega} \varphi_2}{s \models^{\Gamma}_{\omega} \varphi_1 \lor \varphi_2}
\\
  \structrule{s' \models^{\Gamma}_{\omega} \varphi}{s \models^{\Gamma}_{\omega} \M{a}{\varphi}}
\quad\text{if $(s, s') \in E_a$}
\qquad
  \structrule{s \models^{\Gamma}_{\omega} \truefrm \quad \negat (s' \models^{\Gamma}_{\omega} \nnf(\neg\varphi))_{(s, s') \in E_a}}{s \models^{\Gamma}_{\omega} \K{a}{\varphi}}
\end{gather*}
The rule for $s \models^{\Gamma}_{\omega} \K{\ag}{\varphi}$ checks that $s$ is
reachable, but that no counterexample to $\varphi$ can be reached at an
$\ag$-undistinguishable state.

  

Using general rule systems, the solvability of an epistemically guarded
transition system is shifted to computing derivable conclusions.  As for
knowledge"=based programs, it is not obvious from just the rules of a system $R$
whether there are solutions of $\bar{R}(B) = B$ at all, and whether there is a
least one.

\begin{example}
\begin{refpars}
  \item The general rule system
\begin{equation*}
R_0 = \left\{ \structrule{x_1}{x_1}, \structrule{\negat x_1\ \negat x_2}{x_2} \right\}
\quad\text{over}\quad
\{ x_1, x_2 \}
\end{equation*}
has no set of derivable conclusions, since $\bar{R}_0$ has no fixed point; in
particular, $\bar{R}_0(\emptyset) = \{ x_2 \}$ and $\bar{R}_0(\{ x_1 \}) =
\emptyset = \bar{R}_0(\{ x_2 \})$.  $R_0$ also demonstrates that the set of
derivable conclusions of a general rule system $R$ need not coincide with the
least fixed point of the operator $\hat{R} : \powerset U \to \powerset U$ when
transferred from inductive rule systems by now setting $\hat{R}(P) = \{ y \in U
\mid \text{ex.}\ (X, \negat Z)/y \in R\ \text{s.t.}\ X \subseteq P\text{,}\ P
\cap Z = \emptyset \}$: $\mu\hat{R}_0 = \{ x_1 \}$.  On the other hand, in view
of the general rule system for epistemically guarded transition systems $R_0$
can also be rephrased as a knowledge"=based program with a single agent $\ag[a]$
and a single variable $\prop[x] \in \{ 0, 1, 2 \}$, which $\ag[a]$ cannot
observe, started with $\prop[x] = 0$:
\begin{equation*}
\begin{array}[t]{@{}r@{\ }l@{\ }r@{}}
  \textbf{\texttt{if}} & \M{\ag[a]}{\prop[x] = 1} \rightarrowtriangle \prop[x] \gets 1
\\
  \talloblong          & \K{\ag[a]}{(\prop[x] \neq 1 \land \prop[x] \neq 2)} \rightarrowtriangle \prop[x] \gets 2 & \textbf{\texttt{fi}}
\end{array}
\end{equation*}

Similarly, the second rule alone, \ie, the general rule system 
\begin{equation*}
R_1 = \left\{ \structrule{\negat x_1\ \negat x_2}{x_2} \right\}
\quad\text{over}\quad
\{ x_1, x_2 \}
\end{equation*}
does not show any solution, that is a fixed point of $\bar{R}_1$: $x_1$ cannot
be contained in any solution as it cannot be derived, and $x_2$ leads to a
contradiction.  The same holds for
\begin{equation*}
R_2 = \left\{ \structrule{\negat x_1}{x_2},\
              \structrule{\negat x_2}{x_3},\
              \structrule{\negat x_3}{x_1} \right\}
\quad\text{over}\quad
\{ x_1, x_2, x_3 \}
\end{equation*}
where if $x_2$ were derivable, then $x_3$ could not be derived, and thus $x_1$
were derivable, which implies that $x_2$ could not be derived --- and vice
versa;

  \item There may be several solutions of a general rule system, but no least
one:
\begin{equation*}
R_3 = \left\{ \structrule{\negat x_1}{x_3},\ \structrule{\negat x_3}{x_1} \right\}
\quad\text{over}\quad
\{ x_1, x_3 \}
\end{equation*}
has the solutions $\{ x_1 \}$ and $\{ x_3 \}$, but $\emptyset$ is no solution.
It corresponds to the ``variable setting'' knowledge"=based program
\begin{equation*}
\begin{array}[t]{@{}r@{\ }l@{\ }r@{}}
  \textbf{\texttt{if}} & \K{\ag[a]}{\prop[x] \neq 1} \rightarrowtriangle \prop[x] \gets 3
\\
  \talloblong          & \K{\ag[a]}{\prop[x] \neq 3} \rightarrowtriangle \prop[x] \gets 1 & \textbf{\texttt{fi}}
\end{array}
\end{equation*}
of the introduction, see \cref{ex:egts}\cref{it:ex:egts:vs}.  For a general rule
system with an infinite universe and infinitely many rules, consider
\begin{equation*}
R_4 = \left\{ \structrule{\negat x_{n+1}}{x_n} \;\Big|\; n \in \NZ \right\}
\quad\text{over}\quad 
\{ x_n \mid n \in \NZ \}
\end{equation*}
that has the incomparable solutions $\{ x_{2n} \mid n \in \NZ \}$ and $\{
x_{2n+1} \mid n \in \NZ \}$.

  \item Combining the rule systems $R_1$ showing a contradictory rule and $R_3$
with its non"=determined rules, we obtain
\begin{equation*}
R_5 = \left\{ \structrule{\negat x_1}{x_3},\ \structrule{\negat x_3}{x_1},\ \structrule{\negat x_1\ \negat x_2}{x_2} \right\}
\quad\text{over}\quad
\{ x_1, x_2, x_3 \}
\end{equation*}
which has the unique solution $\{ x_1 \}$: if $x_3$ were inferrable, \ie, $x_1$
not inferrable, this would trigger the contradictory rule for $x_2$ (see
\cref{ex:iter}\cref{it:ex:iter:nc}).\qedhere
\end{refpars}
\end{example}

\subsection{Solving General Rule Systems}

The observations and definitions for epistemic must/can transition structures
and constructive interpretation, see \cref{sect:constr}, can now readily be
transferred to a more abstract account for general rule systems.  Define, for a
universe $U$, the set $\powerset^{\pm} U$ as $\{ (P, Q) \in \powerset U \times
\powerset U \mid P \subseteq Q \}$ and the relation ${\subseteq^{\pm}} \subseteq
\powerset^{\pm} U \times \powerset^{\pm} U$ as $(P, Q) \subseteq^{\pm} (P', Q')$
if, and only if, $P \subseteq P'$ and $Q \supseteq Q'$.

\begin{restatable}{lemma}{LEMmustcanpartial}\label{lem:must-can-po}
$(\powerset^{\pm} U, {\subseteq^{\pm}}, \bot^{\pm}_U)$ with $\bot^{\pm}_U =
(\emptyset, U)$ is an inductive partial order.
\end{restatable}

For a general rule system $R$ over $U$ with positive and negative premisses
define the operator $\check{R} : \powerset^{\pm} U \to \powerset^{\pm} U$ that
describes what \emph{must} and what \emph{can} be derived given what is assumed
to be definitely and potentially derivable:
\begin{equation*}
\check{R}(P, Q) = (\stacked{%
  \{ y \in U \mid \text{ex.\ } (X, \negat Z)/y \in R \text{ s.\,t.\ } X \subseteq P,\ Q \cap Z = \emptyset \},\\
  \{ y \in U \mid \text{ex.\ } (X, \negat Z)/y \in R \text{ s.\,t.\ } X \subseteq Q,\ P \cap Z = \emptyset \})}
\end{equation*}
This is well"=defined: if $(P, Q) \in \powerset^{\pm} U$, then $\check{R}(P, Q)
\in \powerset^{\pm} U$, since for $P \subseteq Q$ and each $(X, \negat Z)/y \in
R$ with $X \subseteq P$ and $Q \cap Z = \emptyset$ it holds that $X \subseteq Q$
and $P \cap Z = \emptyset$.  The operator is always monotone:

\begin{restatable}{lemma}{LEMmustcanmonotone}\label{lem:must-can-monotone}
Let $R$ be a rule system over $U$.  If $(P_1, Q_1) \subseteq^{\pm} (P_2, Q_2)$,
then $\check{R}(P_1, Q_1) \subseteq^{\pm} \check{R}(P_2, Q_2)$.
\end{restatable}

As for constructive interpretation, Pataraia's fixed"=point theorem now
guarantees that the monotone operator $\check{R}$ on the inductive partial order
$(\powerset^{\pm} U, {\subseteq^{\pm}}, \bot_U^{\pm})$ has a least fixed point.
Again, it can be ``computed'' by possibly transfinite iterated application of
$\check{R}$ to $\bot_U^{\pm}$.  If, however, $\check{R}$ is even continuous,
then, by Kleene's fixed-point theorem, it suffices to consider all finite
approximations, \ie, $\mu\check{R} = \bigcup^{\pm}_{n \in \NZ}
\check{R}^{n}(\bot_U^{\pm})$; that $\check{R}$ is \emph{continuous} means that
if $\Delta \subseteq \powerset^{\pm} U$ is directed, then $\bigcup^{\pm}
\check{R}(\Delta) = \check{R}(\bigcup^{\pm} \Delta)$.

\begin{restatable}{lemma}{LEMmustcancontinuous}\label{lem:must-can-continuous}
Let $R$ be a rule system over $U$ such that every rule of $R$ shows only
finitely many positive and negative premisses.  Then $\check{R}$ is continuous.
\end{restatable}

The rule system for an epistemically guarded transition system $\Gamma = (S, E,
L, S_0, \mathcal{T})$ over $\ESig$ always shows only finitely many positive
premisses; if for each $s \in S$ and each $\ag \in \Ag$ the set $\{ s' \in S
\mid (s, s') \in E_{\ag} \}$ is finite, then there are also only finitely many
negative premisses, such that the corresponding must/can operator is continuous.

\section{Reasoning About Knowledge"=based Programs}\label{sect:temic}

We have implemented the constructive interpretation of knowledge"=based programs
in the prototypical ``Temporal Epistemic Model Interpreter and Checker''
(\TEMIC\footnote{\url{https://www.dropbox.com/s/ofnatog8ualazdh/temic.zip?dl=0}}).
The tool first computes the least constructive fixed point of a (finite state)
epistemically guarded transition system.  If the least fixed point is decided,
the least solution in terms of epistemic transition structures has been found;
otherwise it is checked whether the re-interpretation using the lower bound of
the undecided least fixed point yields a solution.  If either succeeds,
properties of the resulting model can be checked.  These properties can be
expressed in \CTLK, the combination of the branching ``Computation Tree Logic''
(CTL) and epistemic logic~\cite{lomuscio-penczek:hel:2015}.  What is more, \CTLK
can also be used in \TEMIC for the action guards.  The constructive
interpretation just evaluates each universal quantifier of a CTL formula ---
$\Aop$ for ``on all paths'' --- over the upper bound and each existential
quantifier --- $\Eop$ for ``on some path'' --- over the lower bound.  This adds
the temporal dimension to the domain of application of knowledge"=based
programs.  For the run"=based interpreted systems of Fagin et
al.~\cite{fagin-et-al:dc:1997}, Van der Hoek and
Woolridge~\cite{van-der-hoek-woolridge:spin:2002} and Su~\cite{su:aaai:2004}
provide transformations for linear-time model checking based on local
propositions, though for a fixed set of runs that does not depend on the
evaluation of knowledge guards.  The \CTLK-model checker
MCMAS~\cite{lomuscio-penczek:hel:2015} similarly operates on a fixed,
predetermined model.  In dynamic epistemic logic and its model checker
DEMO~\cite{van-ditmarsch-van-der-hoek-kooi:2008}, the transition structure is
given by epistemic actions.

We first recapitulate briefly \CTLK and then show its constructive evaluation
over epistemic must/can transition structures.  We next describe \TEMIC by means
of the bit transmission problem and the small paradoxical exercise of the
``unexpected examination''; the \TEMIC distribution also contains specifications
for the well"=known problems ``Muddy
Children''~\cite[pp.~93sqq.]{van-ditmarsch-van-der-hoek-kooi:2008} and
``Sum-and-Product''~\cite[pp.~96sq.]{van-ditmarsch-van-der-hoek-kooi:2008}.
Finally, we proceed to an application where \CTLK is also used in the action
guards: the Java memory model.

\subsection{\CTLK}

The \emph{\CTLK-formulæ} over $\ESig$ are defined by the following grammar:
\begin{align*}
  \varphi
&{\;\cln\cln=\;}
  \prop \;\mid\;
  \falsefrm \;\mid\;
  \neg\varphi \;\mid\;
  \varphi_1 \land \varphi_2\;\mid\;
  \K{a}{\varphi} \;\mid\;
  \EX{\varphi} \;\mid\;
  \EG{\varphi} \;\mid\;
  \EU{\varphi_1}{\varphi_2}
\end{align*}
where $p \in P$ and $a \in A$.  The path quantifier $\Eop$ is interpreted as
``there is a path'', the temporal modality $\Xop$ as ``in the next step'',
$\Gop$ as ``always'', and $\Uop$ as ``until''.  We also consider the path
quantifier $\Aop$ for ``all paths'' and the modalities $\Fop$ for ``eventually''
and $\Rop$ for ``release'', such that $\neg\EG{\neg\varphi}$ is abbreviated by
$\AF{\varphi}$ and $\neg\EU{\neg\varphi_1}{\neg\varphi_2}$ by
$\AR{\varphi_1}{\varphi_2}$.  The \emph{satisfaction relation} $M, s \models
\varphi$ of a \CTLK-formula $\varphi$ over $\ESig$ at state $s \in S$ of an
epistemic transition structure $M = (S, E, L, S_0, T)$ over $\ESig$
conservatively extends the satisfaction relation of epistemic formulæ by
\begin{gather*}
  M, s \models \EX{\varphi} \iff \text{ex.\ $s_0, s_1, \ldots \in \Paths(M, s)$ s.\,t.\ $M, s_1 \models \varphi$}
\\
  M, s \models \EG{\varphi} \iff
\begin{array}[t]{@{}l@{}}
  \text{ex.\ $s_0, s_1, \ldots \in \Paths(M,s)$ s.\,t.\ $M, s_i \models \varphi$ f.\,a.\ $i \in \NZ$}
\end{array}
\\
  M, s \models \EU{\varphi_1}{\varphi_2} \iff
\begin{array}[t]{@{}l@{}}
  \text{ex.\ $s_0, s_1, \ldots \in \Paths(M, s)$ and $l \in \NZ$ s.\,t.}\\
  \text{$M, s_i \models \varphi_1$ f.\,a.\ $0 \leq i < l$ and $M, s_l \models \varphi_2$}
\end{array}
\end{gather*}
where $\Paths(M, s)$ denotes all \emph{paths} of $M$, \ie, the infinite state
sequences $s_0, s_1, \ldots \in S$ with $s_0 = s$ and $(s_i, s_{i+1}) \in T$ for
all $i \in \NZ$.  A CTLK"=formula $\varphi$ is \emph{valid} in $M$, written $M
\models \varphi$, if it is satisfied in all initial states, \ie, $M, s_0 \models
\varphi$ for all $s_0 \in \St_0(M)$.

For a direct definition of the satisfaction of CTLK"=formulæ with an $\Aop$, the
existential path quantification for $\Eop$ has to be replaced by universal path
quantification.  As for simple epistemic logic, CTLK including $\AX{\varphi}$,
$\AG{\varphi}$ etc.\ admits a negation normal form (see, \eg,
\cite[pp.~333sq.]{baier-katoen:2008}).  The \emph{constructive satisfaction
  relation} of a CTLK"=formula in negation normal form over an epistemic
must/can transition structure $Y = (S, E, L, S_0, \mathcal{T})$ over $\ESig$ at
a state $s \in \RS(Y_{\nu})$, written $Y, s \models \varphi$, conservatively
extends the constructive satisfaction relation of epistemic formulæ and
interprets $\Eop$ over the lower bound $Y_{\mu}$ and $\Aop$ over the upper bound
$Y_{\nu}$ such that, in particular,
\begin{gather*}
  Y, s \models \EF{\varphi} \iff
\begin{array}[t]{@{}l@{}}
  \text{ex.\ $s_0, s_1, \ldots \in \Paths(Y_{\mu}, s)$ and $i \in \NZ$ s.\,t.\ $Y, s_i \models \varphi$}
\end{array}
\\
  Y, s \models \AF{\varphi} \iff
\begin{array}[t]{@{}l@{}}
  \text{f.\,a.\ $s_0, s_1, \ldots \in \Paths(Y_{\nu}, s)$ ex.\ $i \in \NZ$ s.\,t.\ $Y, s_i \models \varphi$}
\end{array}
\end{gather*}
This extension is well"=defined, as demonstrated in \cref{app:ctlk}.

\subsection{\TEMIC}

\TEMIC is a symbolic model interpreter and checker for epistemically guarded
transition systems using \CTLK.  It is written in Java and uses binary decision
diagrams for state space representation~\cite{jdd}; it also supports bounded
integers and their arithmetic.  Given a specification, \TEMIC first computes the
least constructive fixed point by iterated must/can interpretation.  If this
fixed point is not decided it checks whether another interpretation using the
lower bound of the fixed point yields a solution.  If either succeeds, \TEMIC
proceeds with model checking given properties; these statements can be specified
as \CTLK-formulæ which have to hold in all initial states or as a reachability
query.  Reachable deadlock states without outgoing transitions result in a
warning.

For example, the bit transmission problem of the introduction as formalised in
\cref{ex:egts}\cref{it:ex:egts:bt} can be represented as a \TEMIC specification
as follows:

\begin{lstlisting}[language=TEMIC]
var sbit, ack, rbit, snt : boolean initial (ack | rbit | snt) <-> false;$\rule[-6pt]{0pt}{0pt}$
agent S = { sbit, ack };  agent R = { rbit, snt };
let R_knows_bit = exists bit:boolean . K[R] sbit <-> bit;$\rule[-6pt]{0pt}{0pt}$
action S_sends_bit_ok
guard not K[S] R_knows_bit do rbit := sbit, snt := true; 
action S_sends_bit_failed
guard not K[S] R_knows_bit do ;
action R_sends_ack_ok    
guard R_knows_bit and not K[R] K[S] R_knows_bit do ack := true;
action R_sends_ack_failed
guard R_knows_bit and not K[R] K[S] R_knows_bit do ;
\end{lstlisting}

Constructive interpretation yields in a few milliseconds the decided least fixed
point of \cref{ex:interp-bt}, over which some \CTLK-properties can be checked:

\begin{lstlisting}[language=TEMIC]
check initial EF           R_knows_bit;
check initial EF      K[S] R_knows_bit;
check initial EF K[R] K[S] R_knows_bit;
\end{lstlisting}

The first two are reported to hold, but the last does not since agent $\ag[R]$
cannot gather enough information to be sure that the bit has been received by
agent $\ag[S]$.

For another example, consider the ``unexpected examination''
paradox~\cite[Sect.~4.7, there called ``unexpected
hanging'']{de-haan-et-al:fundinf:2004} (for a detailed account see, \eg,
\cite[Sects.~5.2sq.]{sainsbury:2009}): A class is told that within the next week
there will be an exam, but it will be a surprise.  The class might reason that
the exam cannot happen on Friday, because if there has been no exam up to
Thursday it will not be a surprise on Friday anymore; by backward induction it
might reason that there cannot be a surprise exam in the next week at all.  This
problem statement can be readily expressed as a \TEMIC specification:

\begin{lstlisting}[language=TEMIC]
var day : 0..5 initial day = 0;
var exam : 0..4;
var written : boolean initial written <-> false;$\rule[-6pt]{0pt}{0pt}$
agent P = { day, written };$\rule[-6pt]{0pt}{0pt}$
action act1
guard day < 5 and (day = exam) and (not K[P] day = exam) and not written
do written := true, day := day+1;
action act2
guard day < 5 and (day != exam) do day := day+1;
action stutter
do ;
\end{lstlisting}

Again, constructive interpretation yields in a few milliseconds a decided least
fixed point.  Over this epistemic transition structure we can check that on,
\eg, Wednesday the exam can be written and still is indeed a surprise:

\begin{lstlisting}[language=TEMIC]
check reachable exam = 2 & written;
\end{lstlisting}

For such a reachability check \TEMIC also provides a witness that tells
that \lstinline|act2| is executed twice after which \lstinline|act1| follows.
The following \CTLK-property, however, is not satisfied, as it would have to
hold in all initial states --- and with \lstinline|exam| being 4 the class
cannot be surprised anymore:

\begin{lstlisting}[language=TEMIC]
check initial EF written;
\end{lstlisting}

\subsection{Memory Models}

Memory models regulate the interaction between threads, their caches, and the
main memory~\cite{manson-et-al:jmm:2005}.  The original Java memory model ---
one of the first formal such models --- has been harshly criticised for making
several compiler optimisations impossible and has subsequently been superseded
by a more liberal model~\cite[Ch.~17]{gosling-et-al:2004}.  Keeping strong
guarantees for sequentially consistent, well"=synchronised programs, reorderings
of data"=independent statements or early, ``prescient'' reads from other threads
are allowed for programs with data races.  Still, some limits, like consistency
with data or control flow dependencies or no ``out-of-thin-air'' values, should
be in force~\cite{pugh:jmm,aspinall-sevcik:vamp:2007}.

For example, in the following two"=threaded Java"=like program to the left it
should be possible that both thread"=local registers \code{r1} and \code{r2} are
assigned the value $1$ when reading the global, shared variables \code{x} and
\code{y}: A compiler could reorder the data"=independent statements in both
threads.  This behaviour, however, should be forbidden in the example to the
right, since there is a symmetric data dependence.
\begin{center}\fontsize{8pt}{9.5pt}\selectfont
\begin{tabular}[t]{@{}l@{\quad}||@{\quad}l@{}}
\multicolumn{2}{c}{$\code{x} = \code{y} = 0$}\\
\multicolumn{2}{c}{$\code{r1} = \code{r2} = 0$}
\\\midrule
\code{r1 = x;} & \code{r2 = y;}\\
\code{y = 1;} & \code{x = 1;}
\\\midrule
\multicolumn{2}{c}{$\code{r1} = \code{r2} = 1$?}
\end{tabular}
\qquad\qquad
\begin{tabular}[t]{@{}l@{\quad}||@{\quad}l@{}}
\multicolumn{2}{c}{$\code{x} = \code{y} = 0$}\\
\multicolumn{2}{c}{$\code{r1} = \code{r2} = 0$}
\\\midrule
\code{r1 = x;} & \code{r2 = y;}\\
\code{if (r1 == 1)} & \code{if (r2 == 1)}\\
\quad\code{y = 1;} & \quad\code{x = 1;}
\\\midrule
\multicolumn{2}{c}{$\code{r1} = \code{r2} = 1$?}
\end{tabular}
\end{center}

We want to capture the behaviour of a multi"=threaded (Java) program with a
liberal memory model without having to check all possible compiler
transformations --- the correctness of such transformations would actually
depend on the program semantics including the memory model.  In fact, in the
current Java memory model out-of-order executions have to be justified by other
legal executions.  We interpret these justifications as witnesses in terms of
knowledge"=based programs; our current exposition, however, neglects
synchronisation.  We first represent the state space of a two"=threaded (Java)
program like the ones above by the following \TEMIC declarations:
\begin{lstlisting}[language=TEMIC]
var x, y, r1, r2 : 0..2 initial x = 0 & y = 0 & r1 = 0 & r2 = 0;
var step1, step2 : 1..3 initial step1 = 1 & step2 = 1;$\rule[-6pt]{0pt}{0pt}$
agent t1 = { step1, r1 };  agent t2 = { step2, r2 };
\end{lstlisting}
The thread agents \lstinline|t1| and \lstinline|t2| can only observe their local
registers and their program counters.  The program steps for both threads are
turned into actions like
\begin{lstlisting}[language=TEMIC]
action t1_1 guard step1 = 1 do r1 := x, step1 := step1+1;
action t1_2 guard step1 = 2 do y := 1, step1 := step1+1;
\end{lstlisting}

Additionally, we allow for a ``prescient reading'' of the value $v$ from the
main memory variable $x$ by thread $\theta$ into the local variable $r$ at step
$s$ by the following action:
\begin{lstlisting}[language=TEMIC]
action read$\theta$_$x$_$v$_$r$_$s$
guard step$\theta$ = $s$ and K[$\theta$] (EF ($r$ = 0 & $x$ = $v$) and EF ($r$ = $v$ & $x$ = $v$))
do $r$ := $v$, step$\theta$ := step$\theta$+1;
\end{lstlisting}
The thread $\theta$ can read $v$ from $x$ into $r$ early on if it \emph{knows}
that \emph{there is an execution} where $x$ has value $v$ without dependence on
already setting $r$ to $v$, and, furthermore, that \emph{there is an execution}
where the early setting is confirmed.  The statement \lstinline|r1 = x;| of the
first thread is expanded into three read
actions \lstinline|read1_x_0_r1_1|, \lstinline|read1_x_1_r1_1|,
and \lstinline|read1_x_2_r1_1| plus the plain reading action \lstinline|t1_1|.
With this encoding, \TEMIC reports that for the first example to the left it is
indeed possible to obtain $\code{r1} = \code{r2} = 1$ in the least constructive
fixed point, but that this is impossible for the example to the right.

A more intriguing case is presented by the following two examples: According to
Manson et al.~\cite[pp.~35sq.]{manson-et-al:jmm:2005} (cf.\
also~\cite{aspinall-sevcik:vamp:2007}), the program to the left can result in
$\code{r1} = \code{r2} = \code{r3} = 1$:
\begin{center}\fontsize{8pt}{9.5pt}\selectfont
\begin{tabular}[t]{l@{\quad}||@{\quad}l}
\multicolumn{2}{c}{$\code{x} = \code{y} = 0$}\\
\multicolumn{2}{c}{$\code{r1} = \code{r2} = \code{r3} = 0$}
\\\midrule
\code{r1 = x;} & \code{r3 = y;}\\
\code{if (r1 == 0)} & \code{x = r3;}\\
\quad\code{x = 1;} &\\
\code{r2 = x;} &\\
\code{y = r2;} &
\\\midrule
\multicolumn{2}{c}{$\code{r1} = \code{r2} = \code{r3} = 1$?}
\end{tabular}
\qquad\qquad
\begin{tabular}[t]{l@{\quad}||@{\quad}l@{\quad}||@{\quad}l}
\multicolumn{3}{c}{$\code{x} = \code{y} = 0$}\\
\multicolumn{3}{c}{$\code{r1} = \code{r2} = \code{r3} = 0$}
\\\midrule
\code{r1 = x;} & \code{r2 = x;} & \code{r3 = y;}\\
\code{if (r1 == 0)} & \code{y = r2;} & \code{x = r3;}\\
\quad\code{x = 1;} & &
\\\midrule
\multicolumn{3}{c}{$\code{r1} = \code{r2} = \code{r3} = 1$?}
\end{tabular}
\end{center}
A compiler could see that only $0$ and $1$ are possible for \code{x} and
\code{y} and ``can then replace \code{r2 = x} by \code{r2 = 1}, because either
$1$ was read from \code{x} on line~1 and there is no intervening write, or $0$
was read from \code{x} on line~1, $1$ was assigned to \code{x} on line~3, and
there was no intervening write''; this definite assignment can be used to
transform the last line to \code{y = 1;} which finally can be made the first
action of the first thread, as there are no dependencies.  But the same
transformation is not possible for the program to the right, and there the same
behaviour should be disallowed.  Still, the left program is the result of
inlining the second thread into the first.  Our encoding of the two programs in
\TEMIC confirms these considerations and the witness for the left program indeed
first sets \lstinline|r3| to $1$ and confirms this only in the last step
setting \lstinline|y| to $1$.


\section{Conclusions and Future Work}

We have introduced a must/can analysis for the interpretation of
knowledge"=based programs inspired by the constructive semantics of synchronous
programming languages.  The resulting constructive interpretation provides lower
and upper bounds for the possible executions.  This interpretation has been
shown to be monotone and to yield a least fixed point.  We have also extended
the must/can approach to general rule systems with positive and negative
premisses.  Finally, we have described our tool \TEMIC for constructive
interpretation and temporal"=epistemic model checking over \CTLK and
demonstrated some applications of interpreting knowledge"=based programs
including \CTLK-guards.

Our epistemic logic could be complemented by group
knowledge~\cite[Ch.~6]{fagin-et-al:2003}, like common or distributed knowledge.
The temporal dimension could be extended to ``Linear-Time Logic'' (LTL), and,
more importantly, to include some notion of fairness.
Criteria for ensuring decided least fixed points for the must/can interpretation
beyond synchronicity would be desirable.  This also applies to the solutions of
general rule systems where a comparison with non"=monotone inductive
definitions, see, \eg, \cite{denecker-ternovska:acmtcl:2008}, would be of
interest.  On the other hand, the general constructive approach may be useful to
complement existing intuitionistic approaches to the semantics of synchronous
programming languages~\cite{luettgen-mendler:acmtcl:2002}.  Finally,
the domain of memory models should be covered more comprehensively by
interpreting knowledge"=based programs.


\bibliographystyle{splncs04}
\bibliography{bibliography}

\begin{appendix}
\counterwithin{lemma}{section}

\section{Dependence on the Past and Epistemic Witnesses}\label{app:sect:kbp}

\begin{lemma}\label{lem:reachable}
Let $M_i \in \ETS[\esig](S, E, L, S_0)$ for $1 \leq i \leq 2$.
\begin{enumerate}
  \item\label{it:lem:reachable:1} If $\RT[k](M_1) = \RT[k](M_2)$ for some $k
\geq 0$, then $\RS[k](M_1) = \RS[k](M_2)$.

  \item\label{it:lem:reachable:2} $\RT[k](M_1) = \RT[k](M_2)$ for all $k \geq
0$ if, and only if, $\RT(M_1) = \RT(M_2)$.
\end{enumerate}
\end{lemma}
\begin{proof}
\cref{it:lem:reachable:1}~Let $\RT[k](M_1) = \RT[k](M_2)$ for $k \geq 0$.  If $k
= 0$, then $\RT[0](M_1) = \emptyset = \RT[0](M_2)$ and $\RS[0](M_1) = S_0 =
\RS[0](M_2)$.  If $k > 0$, let $s_1' \in \RS[k](M_1)$ and let $s_1 \in
\RS[k-1](M_1)$ with $(s_1, s_1') \in \Trans(M_1)$.  Then $(s_1, s_1') \in
\RT[k](M_1) = \RT[k](M_2)$ such that $s_1 \in \RS[k-1](M_2)$ and $s_1' \in
\RS[k](M_2)$.  The converse inclusion is symmetric.

\smallskip\noindent%
\cref{it:lem:reachable:2}~If $\RT[k](M_1) = \RT[k](M_2)$ for all $k \geq 0$,
then the definition yields $\RT(M_1) = \bigcup_{0 \leq k} \RT[k](M_1) =
\bigcup_{0 \leq k} \RT[k](M_2) = \RT(M_1)$.~--- Conversely, let $\RT(M_1) =
\RT(M_2)$ hold.  We proceed by induction over $k \geq 0$: For $k = 0$,
$\RT[0](M_1) = \emptyset = \RT(M_2)$.  Now $(s_1, s_1') \in \RT[k+1](M_1)$ if,
and only if, by \cref{it:lem:reachable:1}, $s_1 \in \RS[k](M_1) = \RS[k](M_2)$
and $(s_1, s_1') \in \RT(M_1) = \RT(M_2)$ if, and only if, $(s_1, s_1') \in
\RT[k+1](M_2)$.
\end{proof}

\begin{lemma}\label{lem:limit-ex}
Let $(M_k)_{0 \leq k} \subseteq \ETS[\esig](\ebas)$ such that $\RT[k](M_{k'}) =
\RT[k](M_k)$ for all $k' \geq k \geq 0$, and $M_{\omega} = (\ebas, \bigcup_{0
  \leq k} \RT[k](M_k))$.  Then $\RT[k](M_{\omega}) = \RT[k](M_k)$ for all $k
\geq 0$.
\end{lemma}
\begin{proof}
We proceed by induction on $k \geq 0$: For $k = 0$, $\RT[0](M_{\omega}) =
\emptyset = \RT[0](M_0)$ by definition.  Let the claim hold for $k \geq 0$.
$\RT[k+1](M_{k+1}) \subseteq \RT[k+1](M_{\omega})$ is immediate from the
definition.  For the converse inclusion, let $(s, s') \in \RT[k+1](M_{\omega})$
such that there is an $s \in \RS[k](M_{\omega})$, and hence $s \in \RS[k](M_k)$
by \cref{lem:reachable}\cref{it:lem:reachable:1}, with $(s, s') \in
\Trans(M_{\omega})$, i.e., $(s, s') \in \RT[k'](M_{k'})$ for some $k' > 0$.  If
$k' \leq k+1$, then $(s, s') \in \RT[k'](M_{k'}) = \RT[k'](M_{k+1}) \subseteq
\RT[k+1](M_{k+1})$.  If $k' > k+1$, then $\RS[k](M_k) = \RS[k](M_{k'})$ by
\cref{lem:reachable}\cref{it:lem:reachable:1}.  Thus $s \in \RS[k](M_{k'})$ and
$(s, s') \in \Trans(M_{k'})$ such that $(s, s') \in \RT[k+1](M_{k'}) =
\RT[k+1](M_{k+1})$.
\end{proof}

\begin{lemma}\label{lem:unique}
Let $\Gamma = (\ebas, \mathcal{T})$ be an epistemically guarded transition
system over $\esig$ and let $M_1, M_2 \in \ETS[\esig](\ebas)$ such that all
$\RT[k](M_1) = \RT[k](M_2)$ for all $k \geq 0$.  Then $\interp{\Gamma}{M_1} =
\interp{\Gamma}{M_2}$.
\end{lemma}
\begin{proof}
Let $(s, s') \in \Trans(\interp{\Gamma}{M_1})$.  There is a $(\gact{\varphi}{B})
\in \mathcal{T}$ such that $(s, s') \in B$, $s \in \RS(M_1)$, and $M_1, s
\models \varphi$.  Since $\RT(M_1) = \RT(M_2)$ by
\cref{lem:reachable}\cref{it:lem:reachable:2}, it holds that $s \in \RS(M_1) =
\RS(M_2)$, and thus $M_2, s \models \varphi$.  In particular, $(s, s') \in
\Trans(\interp{\Gamma}{M_2})$.  The reverse inclusion
$\Trans(\interp{\Gamma}{M_2}) \subseteq \Trans(\interp{\Gamma}{M_1})$ is
symmetric.
\end{proof}

\begin{lemma}\label{lem:step}
Let $\Gamma = (\ebas, \mathcal{T})$ be an epistemically guarded transition
system over $\esig$ that depends on the past \wrt $\{ M_1, M_2 \} \subseteq
\ETS[\esig](\ebas)$.
\begin{enumerate}
  \item\label{it:lem:step:1} If $\RT[k](M_1) = \RT[k](\interp{\Gamma}{M_1}) =
\RT[k](\interp{\Gamma}{M_2}) = \RT[k](M_2)$ for some $k \geq 0$, then
$\RT[k+1](\interp{\Gamma}{M_1}) = \RT[k+1](\interp{\Gamma}{M_2})$.

  \item\label{it:lem:step:2} If $\RT(\interp{\Gamma}{M_i}) = \RT(M_i)$ for $1 \leq
i \leq 2$, then $\RT[k](\interp{\Gamma}{M_1}) = \RT[k](\interp{\Gamma}{M_2})$ for all
$k \geq 0$.
\end{enumerate}
\end{lemma}
\begin{proof}
\cref{it:lem:step:1}~Let $\RT[k](M_1) = \RT[k](\interp{\Gamma}{M_1}) =
\RT[k](\interp{\Gamma}{M_2}) = \RT[k](M_2)$ for some $k \geq 0$.  It holds that
$\RT[k](M) \subseteq \RT[k+1](M)$ for every $M \in \ETS[\esig](\ebas)$.  Let
thus $(s, s') \in \RT[k+1](\interp{\Gamma}{M_1}) \setminus
\RT[k](\interp{\Gamma}{M_1})$.  Then, by
\cref{lem:reachable}\cref{it:lem:reachable:1}, $s \in
\RS[k](\interp{\Gamma}{M_1}) = \RS[k](M_1)$ and there is a $(\gact{\varphi}{B})
\in \mathcal{T}$ with $(s, s') \in B$ and $M_1, s \models \varphi$.  Since
$\Gamma$ depends on the past \wrt $\{ M_1, M_2 \}$ and $\RT[k](M_1) =
\RT[k](M_2)$, it holds that $M_2, s \models \varphi$ with $s \in \RS[k](M_2) =
\RS[k](\interp{\Gamma}{M_2})$ and thus $(s, s') \in
\RT[k+1](\interp{\Gamma}{M_2})$.  The converse inclusion is symmetric.

\smallskip\noindent%
\cref{it:lem:step:2}~Let $\RT(\interp{\Gamma}{M_i}) = \RT(M_i)$ for $1 \leq i
\leq 2$.  We proceed by induction on $k \geq 0$: For $k = 0$,
$\RT[0](\interp{\Gamma}{M_1}) = \emptyset = \RT[0](\interp{\Gamma}{M_2})$.  Let
$\RT[k](\interp{\Gamma}{M_1}) = \RT[k](\interp{\Gamma}{M_2})$ hold for $k \geq
0$.  Since $\RT(\interp{\Gamma}{M_i}) = \RT(M_i)$ for $1 \leq i \leq 2$, it
holds using \cref{lem:reachable}\cref{it:lem:reachable:2} that $\RT[k](M_1) =
\RT[k](\interp{\Gamma}{M_1}) = \RT[k](\interp{\Gamma}{M_2}) = \RT[k](M_2)$ such
that $\RT[k+1](\interp{\Gamma}{M_1}) = \RT[k+1](\interp{\Gamma}{M_2})$ by
\cref{it:lem:step:1}.
\end{proof}

\uniqueness*
\begin{proof}
If $\Gamma$ has no solution, then we may choose $\mathcal{M} = \emptyset$; if
$\Gamma$ has a single solution $M$, we may choose $\mathcal{M} = \{ M \}$.~---
For the converse, let $\mathcal{M}$ be given and let $\interp{\Gamma}{M_1} =
M_1$ and $\interp{\Gamma}{M_2} = M_2$ for $M_1, M_2 \in \mathcal{M}$.  By
\cref{lem:step}\cref{it:lem:step:1} and $\RT[0](M_1) =
\RT[0](\interp{\Gamma}{M_1}) = \emptyset = \RT[0](\interp{\Gamma}{M_2}) =
\RT[0](M_2)$, we obtain $\RT[k](M_1) = \RT[k](\interp{\Gamma}{M_1}) =
\RT[k](\interp{\Gamma}{M_2}) = \RT[k](M_2)$ for all $k \geq 0$.  Thus $M_1 =
\interp{\Gamma}{M_1} = \interp{\Gamma}{M_2} = M_2$ by \cref{lem:unique}.
\end{proof}

\existence*
\begin{proof}
We first prove that $\RT[k](M_k) = \RT[k](M_{k'})$ for all $0 \leq k
\leq k'$ by induction over $0 \leq k$.  For $k = 0$ we
have $\RT[0](M_0) = \emptyset = \RT[0](M_{k'})$ for all $0 \leq k'$.  For $k >
0$, from $\RT[k-1](M_{k-1}) = \RT[k-1](M_{k'})$ for all $k-1 \leq k'$ it follows
that $\RT[k-1](M_{k-1}) = \RT[k-1](M_k) = \RT[k-1](\interp{\Gamma}{M_{k-1}}) =
\RT[k-1](\interp{\Gamma}{M_{k'-1}}) = \RT[k-1](M_{k'})$, and thus, by
\cref{lem:step}\cref{it:lem:step:1}, $\RT[k](M_k) =
\RT[k](\interp{\Gamma}{M_{k-1}}) = \RT[k](\interp{\Gamma}{M_{k'-1}}) =
\RT[k](M_{k'})$ for all $k \leq k'$.

\smallskip\noindent%
Now let $\interp{\Gamma}{M_{\omega}} \in \mathcal{M}$.  We next show that
$\RT(\interp{\Gamma}{M_{\omega}}) = \RT(M_{\omega}) =
\RT(\interp{\Gamma}{\interp{\Gamma}{M_{\omega}}})$.  By
\cref{lem:reachable}\cref{it:lem:reachable:2} it suffices to demonstrate that
$\RT[k](\interp{\Gamma}{M_{\omega}}) = \RT[k](M_{\omega}) =
\RT[k](\interp{\Gamma}{\interp{\Gamma}{M_{\omega}}})$ for all $k \geq 0$.  From
this $\interp{\Gamma}{M_{\omega}} =
\interp{\Gamma}{\interp{\Gamma}{M_{\omega}}}$ immediately follows by
\cref{lem:unique}.  By the previous remark, it holds that $\RT[k](M_{k}) =
\RT[k](M_{k'})$ for all $0 \leq k \leq k'$ such that $\RT[k](M_{\omega}) =
\RT[k](M_k)$ for all $k \geq 0$ by \cref{lem:limit-ex}.  We obtain
$\RT[k](\interp{\Gamma}{M_{\omega}}) = \RT[k](M_k)$ by induction on $k \geq 0$:
For $k = 0$, the claim is immediate.  Let $\RT[k](\interp{\Gamma}{M_{\omega}}) =
\RT[k](M_k)$ hold for $k \geq 0$.  From $\RT[k](M_{\omega}) = \RT[k](M_k)$ and
$\RT[k](M_k) = \RT[k](\interp{\Gamma}{M_{k-1}})$ it follows from
\cref{lem:step}\cref{it:lem:step:1} that $\RT[k+1](\interp{\Gamma}{M_{\omega}})
= \RT[k+1](\interp{\Gamma}{M_k}) = \RT[k+1](M_{k+1})$.  By the same argument,
but replacing $\interp{\Gamma}{M_{\omega}}$ by
$\interp{\Gamma}{\interp{\Gamma}{M_{\omega}}}$, we obtain
$\RT[k](\interp{\Gamma}{\interp{\Gamma}{M_{\omega}}}) = \RT[k](M_k)$ by
induction on $k \geq 0$.
\end{proof}

\providing*
\begin{proof}
Let $\gact{\varphi}{B} \in \mathcal{T}$, $M_1, M_2\in \mathcal{M}$, and $k\in
\NZ$ such that $\RT[k](M_1) = \RT[k](M_2)$.  We need to show $M_1, s \models
\varphi \mathrel{\iff} M_2, s \models \varphi$ for all $s \in \RS[k](M_1) \cap
\RS[k](M_2)$.  We proceed by induction on the structure of the epistemic formula
$\varphi$; let $s\in \RS[k](M_1)\cap \RS[k](M_2)$.

\smallskip\noindent%
The cases $\truefrm$, $\falsefrm$, $p$, $\neg\psi$, $\varphi_1 \lor \varphi_2$,
and $\varphi_1 \land \varphi_2$ either hold immediately or follow directly from
the induction hypothesis.

\smallskip\noindent%
\emph{Case $\varphi = \K{a}{\psi}$}: If both $M_1, s \models \K{a}{\psi}$ and
$M_2, s \models \K{a}{\psi}$, the claim holds.  If $M_1, s \not\models
\K{a}{\psi}$, there is, since $M_1$ provides epistemic witnesses, an $s' \in
\RS[k](M_1)$ with $(s, s') \in E_a$ and $M_1, s' \not\models \psi$.  In
particular, using \cref{lem:reachable}\cref{it:lem:reachable:1}, $s' \in
\RS[k](M_1) \cap \RS[k](M_2)$.  By the induction hypothesis, $M_2, s'
\not\models \psi$, and hence $M_2, s \not\models \K{a}{\psi}$.  The converse,
that $M_2, s \not\models \K{a}{\psi}$ implies $M_1, s \not\models \K{a}{\psi}$,
is symmetric.

\smallskip\noindent%
\emph{Case $\varphi = \M{a}{\psi}$}: Analogous to $\K{a}{\psi}$.
\end{proof}

\begin{lemma}
Let $(M_k)_{0 \leq k} \subseteq \ETS(S, E, L, S_0)$ such that $\RT[k](M_{k'}) =
\RT[k](M_k)$ for all $k' \geq k \geq 0$ and let $M_{\omega} = (S, E, L, S_0,
\bigcup_{0 \leq k} \RT[k](M_k))$.  If each $M_k$ provides epistemic witnesses,
then so does $M_{\omega}$.
\end{lemma}
\begin{proof}
Let $k \geq 0$, $s \in \RS[k](M_{\omega})$, and $M_{\omega}, s \not\models
\K{\ag}{\varphi}$.  Then there is a $k' \geq 0$ and an $s' \in
\RS[k'](M_{\omega})$ with $(s, s') \in E_{\ag}$ and $M_{\omega}, s' \not\models
\varphi$.  Thus $M_{k'}, s' \not\models \varphi$, and since $\RT[k'](M_{k'}) =
\RT[k'](M_{\omega})$, also $M_{k'}, s \not\models \K{\ag}{\varphi}$ by
\cref{lem:limit-ex} and \cref{lem:reachable}\cref{it:lem:reachable:1}.  But
$M_{k'}$ provides epistemic witnesses and $s \in \RS[k](M_{k'})$, such that
there is an $s'' \in \RS[k](M_k') \subseteq \RS[k](M_{\omega})$.
\end{proof}

\section{Constructive Interpretation}\label{app:sect:interp}

\LEMposnegconstructive*
\begin{proof}
We proceed by induction over the structure of $\varphi \in \EFrm$.

\smallskip\noindent%
\emph{Case $\varphi = \prop$}: Then $Y, s \pmodels \prop$ iff $p \in L(s)$ iff
$Y, s \not\nmodels \prop$.

\smallskip\noindent%
\emph{Case $\varphi = \falsefrm$}: Then $Y, s \pmodels \falsefrm$ iff $Y, s
\not\nmodels \falsefrm$.

\smallskip\noindent%
\emph{Case $\varphi = \neg\psi$}: Then $Y, s \pmodels \neg\psi$ iff $Y, s
\nmodels \psi$ which implies $Y, s \not\pmodels \psi$ by the induction
hypothesis and thus $Y, s \not\nmodels \neg\psi$.

\smallskip\noindent%
\emph{Case $\varphi = \K{\ag}{\psi}$}: Then $Y, s \pmodels \K{\ag}{\psi}$ iff
$Y, s' \pmodels \psi$ for all $s' \in \RS(Y_{\nu})$ with $(s, s') \in E_{\ag}$
which implies $Y, s \not\nmodels \psi$ for all $s' \in \RS(Y_{\nu})$ with $(s,
s') \in E_{\ag}$ by the induction hypothesis and thus $Y, s \not\nmodels
\K{\ag}{\psi}$ since $\RS(Y_{\mu}) \subseteq \RS(Y_{\nu})$.
\end{proof}

\LEMposnegnnf*
\begin{proof}
We proceed by induction over the structure of $\varphi \in \EFrm$.

\smallskip\noindent%
\emph{Case $\varphi = \prop$}: Then $Y, s \pmodels \prop$ iff $\prop \in L(s)$ iff $Y, s \cmodels \prop = \nnf(\prop)$; and $Y, s \nmodels \prop$ iff $\prop \notin L(s)$ iff $Y, s \cmodels \neg\prop = \nnf(\neg\prop)$.

\smallskip\noindent%
\emph{Case $\varphi = \falsefrm$}: Then $Y, s \pmodels \falsefrm$ iff $Y, s \cmodels \falsefrm = \nnf(\falsefrm)$; and $Y, s \nmodels \falsefrm$ iff $Y, s \cmodels \truefrm = \nnf(\neg\falsefrm)$.

\smallskip\noindent%
\emph{Case $\varphi = \neg\psi$}: Then $Y, s \pmodels \neg\psi$ iff $Y, s
\nmodels \psi$ iff $Y, s \cmodels \nnf(\neg\psi)$ by the induction hypothesis;
and $Y, s \nmodels \neg\psi$ iff $Y, s \pmodels \psi$ iff $Y, s \cmodels
\nnf(\psi) = \nnf(\neg\neg\psi)$ by the induction hypothesis.

\smallskip\noindent%
\emph{Case $\varphi = \varphi_1 \land \varphi_2$}: Then $Y, s \pmodels \varphi_1
\land \varphi_2$ iff $Y, s \pmodels \varphi_1$ and $Y, s \pmodels \varphi_2$ iff
$Y, s \cmodels \nnf(\varphi_1)$ and $Y, s \cmodels \nnf(\varphi_2)$ by the
induction hypothesis iff $Y, s \cmodels \nnf(\varphi_1) \land \nnf(\varphi_2) =
\nnf(\varphi_1 \land \varphi_2)$; and $Y, s \nmodels \varphi_1 \land \varphi_2$
iff $Y, s \nmodels \varphi_1$ or $Y, s \nmodels \varphi_2$ iff $Y, s \cmodels
\nnf(\neg\varphi_1)$ or $Y, s \cmodels \nnf(\neg\varphi_2)$ by the induction
hypothesis iff $Y, s \cmodels \nnf(\neg\varphi_1) \lor \nnf(\neg\varphi_2) =
\nnf(\neg(\varphi_1 \land \varphi_2))$.

\smallskip\noindent%
\emph{Case $\varphi = \K{\ag}{\psi}$}: Then $Y, s \pmodels \K{\ag}{\psi}$ iff
$Y, s' \pmodels \psi$ for all $s' \in \RS(Y_{\nu})$ with $(s, s') \in E_{\ag}$
iff $Y, s \cmodels \nnf(\psi)$ for all $s' \in \RS(Y_{\nu})$ with $(s, s') \in
E_{\ag}$ by the induction hypothesis iff $Y, s \cmodels \K{\ag}{\nnf(\psi)} =
\nnf(\K{\ag}{\psi})$; and $Y, s \nmodels \K{\ag}{\psi}$ iff $Y, s' \nmodels
\psi$ for some $s' \in \RS(Y_{\mu})$ with $(s, s') \in E_{\ag}$ iff $Y, s
\cmodels \nnf(\neg\psi)$ for some $s' \in \RS(Y_{\mu})$ with $(s, s') \in
E_{\ag}$ by the induction hypothesis iff $Y, s \cmodels \M{\ag}{\nnf(\neg\psi)} =
\nnf(\neg\K{\ag}{\psi})$.
\end{proof}

\LEMextepistemic*
\begin{proof}
We proceed by induction over the structure of $\nnf(\varphi)$ in negation normal
form.  The cases $\truefrm$ and $\falsefrm$ are immediate; $\prop$ and $\neg
\prop$ for $\prop \in \Prop$ follow directly from the invariance of the state
labelling of $Y$ and $Y'$; and $\varphi_1 \lor \varphi_2$ and $\varphi_1 \land
\varphi_2$ follow straightforwardly by applying the induction hypothesis.

\smallskip\noindent%
\emph{Case $\nnf(\varphi) = \K{a}{\psi}$}: Let $Y, s \cmodels \K{a}{\psi}$ hold
for an $s \in \RS(Y_{\nu}') \subseteq \RS(Y_{\nu})$.  Then $Y, s' \cmodels \psi$
for all $s' \in \RS(Y_{\nu})$ with $(s, s') \in E_a$.  Thus $Y, s' \cmodels
\psi$ for all $s' \in \RS(Y_{\nu}')$ with $(s, s') \in E_a$, as $\RS(Y_{\nu})
\supseteq \RS(Y_{\nu}')$.  By induction hypothesis, $Y', s' \cmodels \psi$ for
all $s' \in \RS(Y_{\nu}')$ with $(s, s') \in E_a$, and thus $Y', s \cmodels
\K{a}{\psi}$.

\smallskip\noindent%
\emph{Case $\nnf(\varphi) = \M{a}{\psi}$}: Let $Y, s \cmodels \M{a}{\psi}$ hold
for an $s \in \RS(Y_{\nu}') \subseteq \RS(Y_{\nu})$.  Then $Y, s' \cmodels \psi$
for some $s' \in \RS(Y_{\mu})$ with $(s, s') \in E_a$.  Thus $Y, s' \cmodels
\psi$ for some $s' \in \RS(Y_{\mu}')$ with $(s, s') \in E_a$, as $\RS(Y_{\mu})
\subseteq \RS(Y_{\mu}')$.  Since $\RS(Y_{\mu}') \subseteq \RS(Y_{\nu}')$, it
follows by induction hypothesis that $Y', s' \cmodels \psi$ for some $s' \in
\RS(Y_{\mu}')$ with $(s, s') \in E_a$, and thus $Y', s \cmodels \M{a}{\psi}$.
\end{proof}

\PROPconstructivemonotone*
\begin{proof}
Let $(\gact{\varphi}{B}) \in \mathcal{T}$.  Then
$\interp{(\gact{\varphi}{B})}{Y, \mu} \subseteq \interp{(\gact{\varphi}{B})}{Y',
  \mu}$ and $\interp{(\gact{\varphi}{B})}{Y, \nu} \supseteq
\interp{(\gact{\varphi}{B})}{Y', \nu}$ follow from the implications of $Y', s
\cmodels \nnf(\varphi)$ from $Y, s \cmodels \nnf(\varphi)$ for all $s \in
\RS(Y_{\nu}') \supseteq \RS(Y_{\mu}') \supseteq \RS(Y_{\mu})$, and of $Y, s
\not\cmodels \nnf(\neg\varphi)$ from $Y', s \not\cmodels \nnf(\neg\varphi)$ for
all $s \in \RS(Y_{\nu}) \supseteq \RS(Y_{\nu}')$ which are given by
\cref{lem:ext-epistemic}.
\end{proof}

\PROPconstructiveinductive*
\begin{proof}
Let $\Delta \subseteq \EMCTS[\esig](\ebas)$ be directed.  Then $Z$ given by
$Z_{\mu} = \bigcup_{Y \in \Delta} Y_{\mu}$ and $Z_{\nu} = \bigcap_{Y \in \Delta}
Y_{\nu}$ is in $\EMCTS[\esig](\ebas)$, \ie, $\bigcup_{Y \in \Delta} Y_{\mu}
\subseteq \bigcap_{Y \in \Delta} Y_{\nu}$ or, equivalently, $Y_{\mu} \subseteq
Y'_{\nu}$ for all $Y, Y' \in \Delta$.  Indeed, for each $Y, Y' \in \Delta$ there
is, since $\Delta$ is directed, a $Y'' \in \Delta$ with $Y, Y' \sqsubseteq Y''$,
for which then $Y_{\mu} \subseteq Y''_{\mu} \subseteq Y''_{\nu} \subseteq
Y'_{\nu}$ holds.  Moreover, union and intersection form the supremum \wrt
$\subseteq$ and $\supseteq$ for epistemic transition systems, such that
$\bigsqcup \Delta = Z$.  Finally, $\bot_{\esig, \ebas} \sqsubseteq Y$ for all $Y
\in \EMCTS[\esig](\ebas)$.
\end{proof}

\PROPconstructiveunique*
\begin{proof}
$\Gamma$ has the solution $(\mu\Gamma)_{\mu} = (\mu\Gamma)_{\nu} \in
\ETS[\esig](\ebas)$.  For showing uniqueness, let $M \in \ETS[\esig](\ebas)$
satisfy $M = \interp{\Gamma}{M}$.  Let $Y_M = (\ebas, (\Trans(M), \Trans(M)) \in
\ETS[\esig](\ebas)$ such that $Y_M = \interp{\Gamma}{Y_M} \in
\EMCTS[\esig](\ebas)$ (cf.~\cref{sect:decided}).  Since $\mu\Gamma$ is the least
constructive fixed point, it holds that $\mu\Gamma \sqsubseteq Y_M$ and thus
$(\mu\Gamma)_{\mu} \subseteq M \subseteq (\mu\Gamma)_{\nu}$.  As $\mu\Gamma$ is
decided, we obtain $(\mu\Gamma)_{\mu} = M = (\mu\Gamma)_{\nu}$.
\end{proof}

\LEMconstructivesync*
\begin{proof}
We show that $Y_{\mu} = \interp{\Gamma}{Y_{\mu}}$ and $\interp{\Gamma}{Y_{\nu}}
= Y_{\nu}$ from which $Y_{\mu} = Y_{\nu}$ follows from \cref{prop:unique}, since
synchronicity implies the provision of witnesses which in turn implies
dependence on the past, see \cref{sect:unique}.  We only show $Y_{\mu} =
\interp{\Gamma}{Y_{\mu}}$, the other equation is analogous.

By \cref{sect:decided}, it holds that $Y_{\mu} \subseteq
\interp{\Gamma}{Y_{\mu}}$.  Assume for a contradiction to
$\interp{\Gamma}{Y_{\mu}} \subseteq Y_{\mu}$ that there is an $(s, s') \in
\Trans(\interp{\Gamma}{Y_{\mu}}) \setminus \Trans(Y_{\mu})$ and let $k \geq 0$
for $s \in \RS[k](Y_{\mu})$ be minimal.  Then there is a $(\gact{\varphi}{B})
\in \mathcal{T}$ such that $Y, s \not\models \varphi$ and
$\interp{\Gamma}{Y_{\mu}}, s \models \varphi$.  Furthermore, by the minimality
of $k$, $\RT[k](Y_{\mu}) = \RT[k](\interp{\Gamma}{Y_{\mu}})$.  We proceed by
induction over $\nnf(\varphi)$.  The cases $\prop$, $\neg\prop$, $\falsefrm$,
$\truefrm$, $\varphi_1 \land \varphi_2$, and $\varphi_1 \lor \varphi_2$ are
obvious.

\smallskip\noindent%
\emph{Case $\nnf(\varphi) = \K{\ag}{\psi}$}: If $Y, s \not\models
\K{\ag}{\psi}$, then there is some $s'' \in \RS(Y_{\nu})$ with $(s, s'') \in
E_{\ag}$ and $Y, s'' \not\models \psi$.  But $\RS(Y_{\nu}) \subseteq
\RS(\interp{\Gamma}{Y_{\mu}})$ and hence $s'' \in
\RS[k](\interp{\Gamma}{Y_{\mu}}) = \RS[k](Y_{\mu})$ by the synchrony of
$\interp{\Gamma}{Y_{\mu}}$.  Thus, $\interp{\Gamma}{Y_{\mu}}, s'' \not\models
\psi$ and therefore $\interp{\Gamma}{Y_{\mu}}, s \not\models \K{\ag}{\psi}$
would follow by the induction hypothesis.

\smallskip\noindent%
\emph{Case $\nnf(\varphi) = \M{\ag}{\psi}$}: If $\interp{\Gamma}{Y_{\mu}}, s
\not\models \M{\ag}{\psi}$, then there is, by the synchrony of
$\interp{\Gamma}{Y_{\mu}}$, an $s'' \in \RS[k](\interp{\Gamma}{Y_{\mu}}) =
\RS[k](Y_{\mu})$ with $(s, s'') \in E_{\ag}$ and $\interp{\Gamma}{Y_{\mu}}, s''
\models \psi$, from which $Y, s'' \not\models \psi$ and thus $Y, s \not\models
\M{\ag}{\psi}$ would follow by the induction hypothesis.
\end{proof}

\section{Solving General Rule Systems}\label{app:sect:rules}

\LEMmustcanpartial*
\begin{proof}
Let $\Delta \subseteq \powerset^{\pm} U$ be directed.  Then $(\bigcup_{(P, Q)
  \in \Delta} P, \bigcap_{(P, Q) \in \Delta} Q) \in \powerset^{\pm} U$, \ie,
$\bigcup_{(P, Q) \in \Delta} P \subseteq \bigcap_{(P, Q) \in \Delta} Q$ or,
equivalently, $P \subseteq Q'$ for all $(P, Q), (P', Q') \in \Delta$.  Indeed,
for each $(P, Q), (P', Q') \in \Delta$ there is, since $\Delta$ is directed, a
$(P'', Q'') \in \Delta$ with $(P, Q), (P', Q') \subseteq^{\pm} (P'', Q'')$, for
which then $P \subseteq P'' \subseteq Q'' \subseteq Q'$ holds.  Moreover, union
and intersection form the supremum \wrt $\subseteq$ and $\supseteq$, such that
$\bigcup^{\pm} \Delta = (\bigcup_{(P, Q) \in \Delta} P, \bigcap_{(P, Q) \in
  \Delta} Q)$.  Finally, $(\emptyset, U) \subseteq^{\pm} (P, Q)$ for all $(P, Q)
\in \powerset^{\pm} U$.
\end{proof}

\LEMmustcanmonotone*
\begin{proof}
Let $P_1 \subseteq Q_1$, $P_2 \subseteq Q_2$, $P_1 \subseteq P_2$, and $Q_1
\supseteq Q_2$ hold.  Let $(P_i', Q_i') = \check{R}(P_i, Q_i)$ for $1 \leq i
\leq 2$ such that $P_i' \subseteq Q_i'$ for $1 \leq i \leq 2$.  Then $P_1'
\subseteq P_2'$, since if $(X, \negat Z)/y \in R$ such that $X \subseteq P_1$
and $Q_1 \cap Z = \emptyset$, then $X \subseteq P_2$ by $P_1 \subseteq P_2$ and
$Q_2 \cap Z = \emptyset$ by $Q_1 \supseteq Q_2$; and $Q_1' \supseteq Q_2'$,
since if $(X, \negat Z)/y \in R$ such that $X \subseteq Q_2$ and $P_2 \cap Z =
\emptyset$, then $X \subseteq Q_1$ by $Q_1 \supseteq Q_2$ and $P_1 \cap Z =
\emptyset$ by $P_1 \subseteq P_2$.
\end{proof}

\LEMmustcancontinuous*
\begin{proof}
Since $\check{R}$ is monotone \wrt $\subseteq^{\pm}$, it suffices to prove that
$\check{R}(\bigcup^{\pm} \Delta) \subseteq^{\pm} \bigcup^{\pm}
\check{R}(\Delta)$ for $\Delta \subseteq \powerset^{\pm} U$ directed.  Let first
$(X, \negat Z)/y \in R$ such that $X \subseteq \bigcup_{(P, Q) \in \Delta} P$
and $(\bigcap_{(P, Q) \in \Delta} Q) \cap Z = \emptyset$.  Since $X$ and $Z$ are
assumed to be finite and $\Delta$ is directed, there is a $(P', Q') \in \Delta$
with $X \subseteq P'$ and $Q' \cap Z = \emptyset$.  Let now $(X, \negat Z)/y \in
R$ with $X \subseteq \bigcap_{(P, Q) \in \Delta} Q$ and $(\bigcup_{(P, Q) \in
  \Delta} P) \cap Z = \emptyset$.  Then $X \subseteq Q$ and $P \cap Z =
\emptyset$ for all $(P, Q) \in \Delta$.
\end{proof}

\section{Must/Can Interpretation of \CTLK}\label{app:ctlk}

\begin{lemma}
Let $Y \in \EMCTS(S, E, L, S_0)$ and $\varphi$ a CTLK-formula
over $\ESig$.  Then for all $s \in \RS(Y_{\nu})$ it holds that $Y, s
\models \nnf(\varphi)$ implies $Y, s \not\models \nnf(\neg\varphi)$.
\end{lemma}
\begin{proof}
We extend the inductive reasoning of \cref{lem:pos-neg-epistemic} (in the form
of \cref{lem:pos-neg-constructive}) by the temporal cases.

\smallskip\noindent%
\emph{Case $\nnf(\varphi) = \EX{\psi}$}: Then $\nnf(\neg\varphi) =
\AX{\nnf(\neg\psi)}$.  Let $Y, s \cmodels \EX{\psi}$ hold for $s \in
\RS(Y_{\nu})$.  Then $Y, s_1 \cmodels \psi$ for some $s_0, s_1, \ldots \in
\Paths(Y_{\mu}, s)$ such that, in particular, $s \in \RS(Y_{\mu})$.
Since $\Paths(Y_{\mu}, s) \subseteq \Paths(Y_{\nu}, s)$ for $s \in
\RS(Y_{\mu}) \subseteq \RS(Y_{\nu})$, $Y, s_1 \not\models
\nnf(\neg\psi)$ for some $s_0, s_1, \ldots \in \Paths(Y_{\nu}, s)$ by the
induction hypothesis.  Thus $Y, s \not\models \AX{\nnf(\neg\psi)}$.

\smallskip\noindent%
\emph{Case $\nnf(\varphi) = \AX{\psi}$}: Then $\nnf(\neg\varphi) =
\EX{\nnf(\neg\psi)}$.  Let $Y, s \models \AX{\psi}$ hold for $s \in
\RS(Y_{\nu})$.  Then $Y, s_1 \models \psi$ for all $s_0, s_1, \ldots \in
\Paths(Y_{\nu}, s)$.  If $s \in \RS(Y_{\nu}) \setminus \RS(Y_{\mu})$, then there
is no path in $Y_{\mu}$ starting from $s$, and thus $Y, s \not\models
\EX{\nnf(\neg\psi)}$.  Let thus $s \in \RS(Y_{\mu})$.  Since $\Paths(Y_{\mu}, s)
\subseteq \Paths(Y_{\nu}, s)$ for $s \in \RS(Y_{\mu}) \subseteq \RS(Y_{\nu})$,
$Y, s_1 \not\models \nnf(\neg\psi)$ for all $s_0, s_1, \ldots \in
\Paths(Y_{\mu}, s)$ by the induction hypothesis.  Thus again $Y, s \not\models
\EX{\nnf(\neg\psi)}$.

\smallskip\noindent%
All other cases are analogous.
\end{proof}

\begin{lemma}
Let $Y, Y' \in \EMCTS(S, E, L, S_0)$ with $Y \sqsubseteq Y'$ and let $\varphi$
be a CTLK-formula over $\ESig$.  Then $Y, s \models \nnf(\varphi)$ implies $Y',
s \models \nnf(\varphi)$ for all $s \in \RS(Y'_{\nu})$.
\end{lemma}
\begin{proof}
We extend the inductive reasoning of \cref{lem:ext-epistemic} by the temporal cases.

\smallskip\noindent%
\emph{Case $\nnf(\varphi) = \EX{\psi}$}: Let $Y, s \models \EX{\psi}$ hold for
an $s \in \RS(Y_{\nu}') \subseteq \RS(Y_{\nu})$.  Then $Y, s_1
\models \psi$ for some $s_0, s_1, \ldots \in \Paths(Y_{\mu}, s)$, such that, in
particular, $s \in \RS(Y_{\mu})$.  Thus $Y, s_1 \models \psi$ for some
$s_0, s_1, \ldots \in \Paths(Y_{\mu}', s)$, as $\Paths(Y_{\mu}, s) \subseteq
\Paths(Y_{\mu}', s)$.  Since $\RS(Y_{\mu}) \subseteq
\RS(Y_{\mu}') \subseteq \RS(Y_{\nu}')$, it follows by induction
hypothesis that $Y', s_1 \models \psi$ for some $s_0, s_1, \ldots \in
\RS(Y_{\mu}', s)$, and thus $Y', s \models \EX{\psi}$.

\smallskip\noindent%
\emph{Case $\nnf(\varphi) = \AX{\psi}$}: Let $Y, s \models \AX{\psi}$ hold for
an $s \in \RS(Y_{\nu}') \subseteq \RS(Y_{\nu})$.  Then $Y, s_1
\models \psi$ for all $s_0, s_1, \ldots \in \Paths(Y_{\nu}, s)$.  Thus $Y, s_1
\models \psi$ for all $s_0, s_1, \ldots \in \Paths(Y_{\nu}', s)$, as
$\Paths(Y_{\nu}, s) \supseteq \Paths(Y_{\nu}', s)$.  By induction hypothesis,
$Y', s_1 \models \psi$ for all $s_0, s_1, \ldots \in \Paths(Y_{\nu}', s)$, and
thus $Y', s \models \AX{\psi}$.

\smallskip\noindent%
All other cases are analogous.
\end{proof}

\end{appendix}

\end{document}